\documentclass[aps,showpacs,preprintnumbers,amsmath,amssymb,twocolumn,superscriptaddress,floatfix,nofootinbib,10pt,pre]{revtex4-1}

\usepackage{amssymb, amsmath, amsfonts}
\usepackage{amsthm, mathrsfs, amsopn}
\usepackage{bm,bbm}
\usepackage{dcolumn}
\usepackage{graphicx}
\usepackage{mathtools}
\usepackage{subfigure}
\usepackage{tabularx}
\usepackage{tikz}
\usepackage{xkcdcolors}
\usepackage{hyperref}
\usepackage{amsthm}
\hypersetup{colorlinks=true, citecolor=orange, urlcolor=blue, linkcolor=magenta}

\newtheorem{theorem}{Theorem}
\newtheorem{lemma}{Lemma}
\begin{document}

\title{Global synchronization in time-varying directed higher-order networks}

\author{Wilfried Segnou}
\email{wilfried.segnou@unamur.be}

\affiliation{Department of Mathematics and Namur Institute for Complex Systems, naXys, University of Namur, Belgium}

\affiliation{Physics and Applied Mathematics Unit, Indian Statistical Institute, 203 B. T. Road, Kolkata 700108, India}

\author{Timoteo Carletti}
\email{timoteo.carletti@unamur.be}
\affiliation{Department of Mathematics and Namur Institute for Complex Systems, naXys, University of Namur, Belgium}

\author{Dibakar Ghosh}
\email{diba.ghosh@gmail.com}
\affiliation{Physics and Applied Mathematics Unit, Indian Statistical Institute, 203 B. T. Road, Kolkata 700108, India}

\begin{abstract}
Global synchronization has been largely investigated in static higher-order networks and little attention has been brought to directed, temporal higher-order networks. We hereby make a step in this direction by studying the emergence of global synchronization in periodically fast switching systems with directed higher-order interactions. By working in the framework of the Master Stability Function, we show that the linearized dynamics includes two time-dependent Jacobian matrices, both evolving on the timescale of the synchronous solution, together with a generalized Laplacian evolving on a fast switching timescale, resulting from the hypergraph time evolution. To address this two-timescale structure, we establish a fast-switching stability criterion for directed balanced higher-order structures. The latter condition allows to decouple the parallel and transverse perturbations to the synchronization manifold. By relying on the average theory for dynamical systems, we are able to prove exponential stability of global synchronization in the fast switching higher-order network. The proposed theory is then illustrated by using a higher-order Stuart--Landau model, which provides a natural benchmark for the theory. Directionality can move Laplacian eigenvalues into instability regions off the real axis, hence destabilizing synchronization; our results nevertheless reveal that temporal switching can counteract this destabilizing effect and restore global synchronization when the directed structures alternate sufficiently rapidly, even though the instantaneous directed higher-order structures are individually unstable.
\end{abstract}

\maketitle

\section{Introduction}
\label{sec:Introduction}

Network science provides a unified framework for modeling complex systems, the latter being composed of many interacting units and used to study many systems, ranging from biological and ecological populations to social, technological, and engineered systems~\cite{barabasi2003scale,albert2002statistical,newman2018networks,lietz2018watts,strogatz2001exploring}. By representing individual units as nodes and their interactions as edges, network models help explain how local interactions generate global behaviors that cannot be inferred from the knowledge of the isolated units alone. Among these collective phenomena, synchronization has received particular attention in the field of networked dynamical systems. In a broad sense, synchronization describes the spontaneous emergence of coordinated dynamics among interacting units, despite differences in their initial states or individual evolution~\cite{becchetti2020consensus,dolfin2015modeling,valentini2017achieving}.

Of particular interest is global synchronization, in which identical nodes asymptotically converge towards a common trajectory; let us observe that ``\textit{global}'' refers to synchronization across the entire network and does not imply global asymptotic stability, for this reason the latter is also named complete synchronization in the literature. It is very often studied in the case the single units are, regular or irregular, oscillators and it is used to describe numerous natural and technological systems, by including the coherent flashing of fireflies
~\cite{buck1976synchronous,buck1988synchronous,sarfati2021self}, neuronal and cardiac activity~\cite{sato2009synchronization,winfree1988time}, chemical oscillations
~\cite{toiya2010synchronization,taylor2009dynamical}, coupled laser arrays
~\cite{terry1999synchronization,deshazer2001detecting}, power grids~\cite{motter2013spontaneous,rohden2012self}, and wind turbines~\cite{kouonang2026synchronization,van2024synchronized}.
For this reason, scholars have developed several tools among which the Master Stability Function (MSF)~\cite{fujisaka1983stability,pecora1998master} has played a central role. The latter relates the linear stability of the synchronous state to the local dynamics, the coupling function and to the spectrum of the network Laplacian, hence synchronization strongly depends on structural properties such as connectivity, heterogeneity, modularity, and
directionality~\cite{arenas2008synchronization,townsend2020dense}.

The dependence of synchronization on connectivity naturally raises
the question of how collective coherence is affected when the
interaction structure itself evolves in time. Several extensions of
the MSF have therefore been developed to
account for temporal variations of the Laplacian spectrum and its
eigenvectors~\cite{carletti2022theory,zhang2021designing}. Nevertheless, analytically predicting stability of global synchronization in a general time-varying network remains a difficult task,
since the complete temporal evolution of both the eigenvalues and the
associated eigenbasis may be required. We refer the interested reader to the review work~\cite{ghosh2022synchronized} for more details about synchronization in temporal networks. To overcome the analytical difficulties and because of their technological applications to switched systems~\cite{roberson2005control,bhowmick2012experimental,porfiri2009global}, scholars have studied the case where the coupling network periodically switches between a finite number of configurations, a priori known. When the switching is sufficiently fast, averaging theory can be invoked to relate the temporal dynamics to the one generated by an effective averaged structure
~\cite{maggia2020higher,sanders2007averaging,petit2017theory}.
However, closeness between the trajectories of the temporal and
averaged systems does not by itself guarantee the asymptotic stability
of the synchronous solution even if the latter exhibits synchronization. 

Over the past three decades, increasing attention has been devoted to the stability of slow-fast linear systems. Rigorous averaging methods have therefore been developed to relate the asymptotic stability of such systems to that of a reduced averaged model, obtained by averaging over the fast temporal scale~\cite{choi2005stability,liberzon2022stability,shim2025stability}. In the context of network dynamics, a stronger result was established in~\cite{stilwell2006sufficient}, showing that exponential stability of the synchronization manifold of the averaged network can be transferred to the original time-varying network provided that the switching is sufficiently fast. In other words, if perturbations transverse to the synchronization manifold decay in the averaged network, then they also decay in the rapidly switching network, so that the synchronous state remains stable below a sufficiently small switching threshold. To the best of our knowledge, the theory has been developed by assuming a linear coupling, reducing thus the linearized system to a sum of two terms, the slow time-varying Jacobian of the ``reaction'' part, and the fast time-varying Laplace matrix. By considering a nonlinear coupling, we will obtain a more complicated structure exhibiting the product of slow and fast matrices. This is the more general framework we will consider in the following.

A large amount of literature about synchronization of coupled oscillators, assumed that interactions occur exclusively between pairs of nodes, whereas many collective
processes involve the simultaneous participation of several units, i.e., many-body interactions. Recently, researchers have shown that  hypergraphs provide a natural representation of the latter by allowing a single hyperedge to connect more than two
nodes~\cite{battiston2021physics,boccaletti2023structure,
zhang2024deeper}. The dynamics associated with these structures naturally involves
multivariate nonlinear coupling functions, whose arguments represent the states
of the different nodes participating in the same hyperedge. Directionality further distinguishes nodes that transmit information from those receiving it, thereby
describing asymmetric group interactions that cannot be
represented by undirected hypergraphs
~\cite{gallo2022synchronization,dorchain2024impact,li2024enhancing}. By exploiting their symmetry properties, scholars have been able to highlight their impact on problems related to the stability of global synchronization~\cite{lucas2020multiorder,gallo2022synchronization,gambuzza2021synchronization}. Although the vast majority of previous studies have focused on static underlying higher-order structures, some works have extended synchronization analysis to fast-switching higher-order networks~\cite{anwar2023synchronization,rakshit2020intralayer}. While important analytical and numerical results have been obtained, a general rigorous proof establishing the asymptotic stability of synchronization manifold in fast switching higher-order structures remains, to the best of our knowledge, unavailable. 

In fast switching directed hypergraphs, one can again resort to a linear stability analysis about the synchronous trajectory, facing thus to a linearized system involving a time-dependent generalized Laplacian, inherited from the switching hypergraph, and time-dependent local and coupling Jacobian matrices evaluated along the synchronous solution. For the sake of pedagogy and to make a step forward, we hereby assume the temporal directed hypergraph to undergo fast periodic switching among a finite set of fixed configurations, allowing us to easily compute the average of the time-dependent generalized Laplacian over one switching period. We are then able to establish a stability criterion that reduces the analysis of the linearized system involving products and sum of terms evolving on two distinct time scales, to the one of an averaged system depending only on the slow time scale, provided the switching frequency of the directed hypergraph is sufficiently large. In order to isolate perturbations transverse to the synchronization manifold in both the averaged and the temporal linearized systems, we assume the directed hypergraph to be balanced at all times.

For the sake of pedagogy, we introduce a higher-order Stuart-Landau (SL) model to benchmark the proposed theory. The SL oscillator results to be the normal form of any dynamical system close to a supercritical Hopf bifurcation and provides a generic amplitude--phase description of self-sustained oscillations
~\cite{kuramoto1984chemical,garcia2012complex,nakao2014complex,van2021theory,di2017benjamin}; for this reason the results hereby presented still hold true beyond the actual example. In contrast to
phase-only descriptions, SL retains amplitude variations while its
limit-cycle solution remains analytically accessible, making the
model particularly suitable for an analytical investigation of
synchronization beyond purely phase dynamics; indeed the Jacobian associated to the local ``reaction'' part becomes time independent, even if evaluated on the synchronous solution. To further simplify the analysis, we assume to deal with a nonlinear multivariate resonant 
coupling~\cite{rspa.2025.0907}. This condition removes the explicit
time dependence of the coupling Jacobian so that
the temporal dependence of the linearized dynamics is entirely carried
by the generalized fast switching Laplacian. Furthermore we show that, when the resonance condition is not satisfied, we can resort to the use of Floquet theory to study the stability of the synchronous solution.

Our analysis reveals a nontrivial interplay between directionality and
temporality. By using the MSF, we show that directionality can shift Laplacian eigenvalues into instability regions, destabilizing synchronization otherwise stable on the corresponding undirected structure. Temporal switching can counteract this effect: if the averaged transverse dynamics is stable, sufficiently rapid alternation between individually destabilizing directed configurations restores local stability of global synchronization. We illustrate this mechanism by using directed circulant hyperrings, whose explicit spectra allow an analytical comparison of the instantaneous and averaged structures.


\section{The general framework}
\label{sec:problem}

Let us consider a population of $N$ identical dynamical systems whose dynamics, once isolated each other, is governed by
\begin{equation}
\label{eq:isolated_general_system}
\frac{d{\bm x}_{\ell}}{dt}
=
\mathbf{f}({\bm x}_{\ell})\, ,
\quad
\ell=1,\ldots,N,
\end{equation}
where ${\bm x}_{\ell}\in\mathbb{R}^{n}$ denotes the state of node
$\ell$ and
$\mathbf{f}:\mathbb{R}^{n}\rightarrow\mathbb{R}^{n}$
is a nonlinear vector field. We assume that units interact each other through directed higher-order structures whose connectivity  evolves in time. To introduce directionality in higher-order interactions, let us consider $d+1$ distinct nodes belonging to a given hyperedge and divide them into two disjoint groups, one made by $m$ nodes and the second one by $q=d+1-m$. The latter nodes form the \emph{tail} and are characterized by the fact that they jointly transmit information only towards nodes in the remaining group, the \emph{head};  nodes in the head interact each other but not with those in the tail. This defines an $m$-directed $d$-hyperedge.  Hence, the interaction defines a preferred direction from the tail to the head~\cite{gallo2022synchronization,dorchain2024impact}. Such structure can be represented by an adjacency tensor
$\mathbf{A}^{(d,m)}$, where the indexes denote the order, $d$, and
directionality, $m$; if the nodes $\ell, i_1,\cdots, i_{m-1}$ are in the head and $j_1,\cdots, j_q$ in the tail, we have $\mathbf{A}^{(d,m)}_
{(\ell i_1\cdots i_{m-1})(j_1\cdots j_q)}=1$. By definition nodes in the head interact each other and thus they are indistinguishable each other, the same holds true for nodes in the tail, hence the tensor is invariant under permutations performed separately within these two groups, namely
\begin{equation}
\mathbf{A}^{(d,m)}_
{\pi_1(\ell\,i_1,\ldots,i_{m-1})
\pi_2(j_1,\ldots,j_q)} =
\mathbf{A}^{(d,m)}_
{(\ell i_1\cdots i_{m-1})(j_1\cdots j_q)}\, ,
\end{equation}
for every admissible permutation $\pi_1$ of the head indices and
$\pi_2$ of the tail indices. Let us observe that once the hypergraph varies in time we will assume the latter equation to be valid for all $t>t_0=0$ and we will denote the tensor by $\mathbf{A}^{(d,m)}_{\pi_1(\ell\,i_1,\ldots,i_{m-1})\pi_2(j_1,\ldots,j_q)} (t)$.

Although this construction accommodates arbitrary values of $m$, in
the remainder of this work we restrict ourselves to the case
$m=1$. Thus, every directed hyperedge contains a single node in its
head and $d$ nodes in its tail, as we will see later this simplifies the analysis. Furthermore, since no ambiguity will arise, we shall
henceforth refer to such interactions simply as \emph{directed}
$d$-hyperedges rather than $1$-directed $d$-hyperedges. The adjacency tensor is therefore denoted simply by $\mathbf{A}^{(d)}$, with entries $\mathbf{A}^{(d)}_{\ell j_1\cdots j_d}=p$
whenever $\ell$ belongs to the head of the directed hyperedge,
$j_1,\ldots,j_d$ belong to its tail, and $p>0$ is the corresponding
interaction weight. The pairwise directed network case is recovered for $d=1$. 

For the sake of pedagogy and to simplify a bit more the analysis, we focus on $d$-uniform directed hypergraphs, so that all interactions involve exactly $d+1$ nodes at any given time. The dynamics of node $\ell$ can then be written in the general form by
\begin{equation*}
\frac{d{\bm x}_{\ell}}{dt}
=
\mathbf{f}({\bm x}_{\ell})
+
\mu
\sum_{j_1,\ldots,j_d=1}^{N}
\mathbf{A}^{(d)}_{\ell j_1\cdots j_d}(t)
\,
\mathbf{g}^{(d)}
\left(
{\bm x}_{\ell},
{\bm x}_{j_1},\ldots,{\bm x}_{j_d}
\right)\, ,
\end{equation*}
where $\mathbf{g}^{(d)}:\mathbb{R}^{n(d+1)}\rightarrow\mathbb{R}^{n}$ describes the nonlinear higher-order interaction and $\mu$ denotes its coupling strength. In this work we are interested in hypergraphs whose time evolution is $T$-periodic, for some $T>0$, i.e., $\mathbf{A}^{(d)}_{\ell j_1\cdots j_d}(t)=\mathbf{A}^{(d)}_{\ell j_1\cdots j_d}(t+T)$ for all $t$ and all the indexes. Moreover we assume the time evolution of the hypergraph structure to be faster than the natural time scale of the isolated system~\eqref{eq:isolated_general_system}, this assumption can be encoded into the following equation
\begin{equation}
\label{eq:general_problem_model}
\frac{d{\bm x}_{\ell}}{dt}
=
\mathbf{f}({\bm x}_{\ell})
+
\mu
\sum_{j_1,\ldots,j_d=1}^{N}
\mathbf{A}^{(d)}_{\ell j_1\cdots j_d}(t/\varepsilon)
\,
\mathbf{g}^{(d)}
\left(
{\bm x}_{\ell},
{\bm x}_{j_1},\ldots,{\bm x}_{j_d}
\right)\,,
\end{equation}
where
$\varepsilon>0$ controls the timescale of the structural evolution. In particular, small values of
$\varepsilon$ correspond to a rapidly varying hypergraph. 

Before to analyse the system~\eqref{eq:general_problem_model} in quest of possible global synchronization we need to introduce the notion of {\em balanced directed uniform hypergraphs}. Let us consider the generalized in-degree
\begin{equation}
\label{eq:genindeg}
 k_\ell^{\mathrm{in}}
=
\frac{1}{d!}
\sum_{j_1,\ldots,j_d=1}^{N}
\mathbf{A}^{(d)}_{\ell j_1\cdots j_d}\, ,
\end{equation}
the latter  denoting the total weight of hyperedges containing node $\ell$ in their head. We also consider the total weight of directed hyperedges containing node $\ell$ in the head and node $j$ in the tail
\begin{equation}
\label{eqkinellj}
 k_{\ell j}
=
\frac{1}{(d-1)!}
\sum_{j_2,\ldots,j_d=1}^{N}
\mathbf{A}^{(d)}_{\ell jj_2\cdots j_d}\, .
\end{equation}
With those definitions we can write the generalized Laplacian matrix~\cite{gallo2022synchronization,dorchain2024impact} as
\begin{equation}
\label{eq:Lapd}
    \mathbf{L}^{(d)}_{\ell j}
=
\begin{cases}
-d!\;k_\ell^{\mathrm{in}},
& \ell=j,\\[4mm]
(d-1)!\;k_{\ell j},
& \ell \ne j\, .
\end{cases}
\end{equation}
Let us observe that the latter matrix has been called $\check{\;\mathbf{L}}^{(d,1)}$ in~\cite{dorchain2024impact}, however to lighten the notations we here prefer to use a simplified notation. Finally, we define the out-degree as the total weight of hyperedges that have node $\ell$ in tail 
\begin{equation}
\label{eq:kout}
 k_\ell^{\mathrm{out}}
=
\frac{1}{(d-1)!}
\sum_{j_1,j_2,\ldots,j_d=1}^{N}
\mathbf{A}^{(d)}_{j_1 \ell j_2\cdots j_d}\, .
\end{equation}

A directed uniform hypergraph is said to be \textit{balanced} if, for each node $\ell$, the total weight of the interaction it determines as tail node, is equal to the total weight of the interaction it receives as a head. Since each directed $d-$hyperedge involves $d$ tail nodes jointly acting on its head, balance condition reads 
\begin{equation}
\label{eq:balanced}
    k_\ell^{\mathrm{out}} = d\;k_\ell^{\mathrm{in}}, \quad \ell=1,\dots,N\, .
\end{equation}

Based on this definition, the matrix $\mathbf{L}^{(d)}$ is balanced, i.e., the outgoing degree equals the incoming one. This implies that the vector $\mathbf{u}=(1,\dots,1)^\top/\sqrt{N}$ lies in both the kernel of the matrix $\mathbf{L}^{(d)}$ and its transpose

\begin{equation}
    \mathbf{L}^{(d)}\mathbf{u}=\mathbf{u}^\top\mathbf{L}^{(d)}=0\, .
\end{equation}

Our interest is the emergence and stability of global synchronization. Let ${\bm x}_s(t)$ be a solution of the isolated dynamics Eq.~\eqref{eq:isolated_general_system}, then global synchronization corresponds to the existence and stability of the invariant manifold
\begin{equation}
\label{eq:general_sync_manifold}
{\bm x}_1(t)= {\bm x}_2(t)=\cdots= {\bm x}_N(t)= {\bm x}_s(t)\, .
\end{equation}
For this manifold to be a solution of the coupled problem Eq.~\eqref{eq:general_problem_model},  the interaction term must vanish whenever all its arguments coincide. A
natural way to guarantee this property is to assume a \textit{diffusive-like}
higher-order coupling ~\cite{gambuzza2021synchronization,gallo2022synchronization,lucas2020multiorder}, namely
\begin{equation}
\label{eq:general_diffusive_coupling}
\begin{aligned}
\mathbf{g}^{(d)}
\left(
{\bm x}_{\ell},
{\bm x}_{j_1},\ldots,{\bm x}_{j_d}
\right)
&=
\mathbf{h}^{(d)}
\left(
{\bm x}_{j_1},\ldots,{\bm x}_{j_d}
\right)+
\\
&-
\mathbf{h}^{(d)}
\left(
{\bm x}_{\ell},\ldots,{\bm x}_{\ell}
\right)\, ,
\end{aligned}
\end{equation}
where
$\mathbf{h}^{(d)}:\mathbb{R}^{nd}\rightarrow\mathbb{R}^{n}$.
Consequently, $
\mathbf{g}^{(d)}
(\mathbf{x},\ldots,\mathbf{x})
=
\mathbf{0},
$ 
and the synchronous trajectory of the isolated system is also a solution of the coupled dynamics, independently of the instantaneous hypergraph configuration. 

Let us investigate the linear stability of the synchronous solution. By
introducing small heterogeneous perturbations around the synchronous
trajectory
\begin{equation*}
 {\bm x}_{\ell}(t)
=
{\bm x}_{s}(t)
+
\delta{\bm x}_{\ell}(t),
\text{ with }
\|\delta{\bm x}_{\ell}\|\ll 1,
\end{equation*}
the linearized dynamics yields~\cite{gambuzza2021synchronization,gallo2022synchronization,
dorchain2024impact,lucas2020multiorder}
\begin{equation}
\label{eq:general_linear_node}
\frac{d}{dt}\delta{\bm x}_{\ell}
=
\mathbf{J}_0({\bm x}_{s}(t))\,
\delta{\bm x}_{\ell}
+
\mu\,
\sum_{j=1}^{N}
\mathbf{L}^{(d)}_{\ell j}(t/\varepsilon)
\mathbf{J}^{(d)}_{\mathbf h}({\bm x}_{s}(t))\delta{\bm x}_{j}\, ,
\end{equation}
for all node $\ell\in\{1,\dots,N\}$.  The diffusive-like assumption allows to introduce the generalized Laplacian matrix defined in Eq.~\eqref{eq:Lapd} in the latter equation. The matrix
\begin{equation}
\label{eq:J0_general}
\mathbf{J}_0({\bm x}_{s}(t))
=
\left.
\frac{\partial\mathbf{f}}{\partial\mathbf{x}}
\right|_{\mathbf{x}={\bm x}_{s}(t)}\, ,
\end{equation}
is the Jacobian of the local dynamics evaluated along the synchronous
trajectory and similarly
\begin{equation}
\label{eq:partial_Jh}
\mathbf{J}^{(d)}_j({\bm x}_{s}(t))
=
\left.
\frac{\partial
\mathbf{h}^{(d)}
({\bm x}_{1},\ldots,{\bm x}_{d})}
{\partial{\bm x}_{j}}
\right|_{
{\bm x}_{1}=\cdots={\bm x}_{d}={\bm x}_{s}(t)}\, ,
\end{equation}
for all $j=1,\ldots,d$, is the Jacobian of the higher-order coupling function with respect to its
$j$-th argument. The coupling Jacobian appearing in
Eq.~\eqref{eq:general_linear_node} is
\begin{equation}
\label{eq:Jh_general}
\mathbf{J}^{(d)}_{\mathbf h}({\bm x}_{s}(t))
=
\sum_{j=1}^{d}
\mathbf{J}^{(d)}_{j}({\bm x}_{s}(t))
=
\left.
\sum_{j=1}^{d}
\frac{\partial\mathbf{h}^{(d)}}
{\partial{\bm x}_{j}}
\right|_{
{\bm x}_{1}=\cdots={\bm x}_{d}={\bm x}_{s}(t).}\, 
\end{equation}

Finally, by introducing the global perturbation vector
$\delta\mathbf{X}
=
\left(
\delta{\bm x}_{1}^{\top},
\ldots,
\delta{\bm x}_{N}^{\top}
\right)^{\top}$, the complete linearized dynamics takes the compact form
\begin{equation}
\label{eq:general_linear_compact}
\frac{d}{dt}\delta\mathbf{X}
=
\left[
\mathbf{I}_{N}\otimes\mathbf{J}_{0}({\bm x}_{s}(t))
+
\mu\,
\mathbf{L}^{(d)}(t/\varepsilon)
\otimes
\mathbf{J}^{(d)}_{\mathbf h}({\bm x}_{s}(t))
\right]
\delta\mathbf{X}\, .
\end{equation}
Notice here that the restriction to a single head node, i.e., $m=1$, allowed us to derive the latter equation. Indeed for a general $m$-directed interaction with
$m>1$, linearization around the synchronous manifold generally introduces two Laplace matrices~\cite{dorchain2024impact}, one for which both nodes belong to the head, and the other where one node lies in the head and the other in the tail as in Eq.~\eqref{eq:Lapd}. The first matrix is symmetric, and the second is non-symmetric; thus, the two matrices need not share the same eigenbasis. As a consequence the machinery of MSF cannot be applied 
without imposing additional condition on the coupling Jacobians
~\cite{dorchain2024impact}. We therefore do not address the case $m>1$ which represents a natural extension of this work. 

Equation~\eqref{eq:general_linear_compact} highlights the main
difficulty addressed here. Even in the directed case with a
single head node, both the generalized Laplacian
$\mathbf{L}^{(d)}(t/\varepsilon)$ and the coupling
Jacobian $\mathbf{J}^{(d)}_{\mathbf h}({\bm x}_{s}(t))$ may depend explicitly on
time. The linearized coupling therefore contains quantities evolving with different time scales and this fact introduces a major difficulty in the analysis of the stability of the synchronous solution. 

Let us note that Eq.~\eqref{eq:general_linear_compact} can be written as 
\begin{equation}
 \label{eq:compactform}
 \frac{d}{dt}\delta\mathbf{X}
=
\left[
{\bf A}(t)
+
{\bf E}(t/\varepsilon){\bf B}(t)\right]
\delta\mathbf{X}\,,
\end{equation}
for suitable time dependent matrices ${\bf A}$, ${\bf E}$ and ${\bf B}$. As mentioned in the Introduction, previous studies have investigated the stability of non-autonomous linear systems involving slow-fast time scales by relating their exponential stability properties to those of a corresponding averaged system evolving only on the slow time scale. Roughly speaking, the main idea is to project Eq.~\eqref{eq:compactform} onto the direction transverse to the synchronization manifold, and then to infer exponential stability along this direction, from the corresponding averaged network. Starting from the theory developed~\cite{stilwell2006sufficient}, those ideas have subsequently been employed in a variety of synchronization problems~\cite{jardon2024persistent,frasca2008synchronization,parastesh2022blinking}. However, to the best of our knowledge, a complete analysis has been provided only under the assumption of ${\bf B}(t)$ to be equal to the identity matrix, ${\bf E}$ has a zero average, and $\bf A(t)$ is a uniformly Hurwitz matrix for every time~\cite{liberzon2022stability}. Similar assumptions are made in~\cite{shim2025stability} once considering the case ${\bf B}(t)\ne {\bf I}$.

The latter assumptions are not general enough to cover many relevant cases, indeed the matrix $\mathbf{J}_{0}({\bm x}_{s}(t))$ is not necessarily uniformly Hurwitz for every time, and the average of the generalized Laplacian matrix is not vanishing. By assuming ${\bf B}(t)={\bf I}$ relies on the hypothesis of linear coupling function; we herein generalize this approach by considering nonlinear interactions and thus returning a time dependent matrix ${\bf B}(t)$.

By denoting $\langle \mathbf{A}^{(d)}\rangle$ the adjacency tensor associated to the averaged directed hypergraph, that is, for all $t\ge t_0$,
\begin{equation}
\label{eq:averadgetensor}
 \langle \mathbf{A}^{(d)}\rangle = \frac{1}{T}\int_{t}^{t+T} \mathbf{A}^{(d)}(\tau)d\tau\, ,  
\end{equation}
we wish to determine a stability criterion for the stability of the synchronous solution of Eq.~\eqref{eq:general_problem_model} based on the following averaged system 
\begin{equation}
\label{eq:general_problem_model_average}
\frac{d\tilde{{\bm x}}_{\ell}}{dt}
=
\mathbf{f}(\tilde{{\bm x}}_{\ell})
+
\mu
\sum_{j_1,\ldots,j_d=1}^{N}
\langle\mathbf{A}^{(d)}_{\ell j_1\cdots j_d}\rangle
\,
\mathbf{g}^{(d)}
\left(
\tilde{{\bm x}}_{\ell},
\tilde{{\bm x}}_{j_1},\ldots,\tilde{{\bm x}}_{j_d}
\right)\,.
\end{equation}

Theorem~\ref{thm:main} will be devoted to prove such claim. The following lemma, whose proof is presented in Appendix~\ref{App:linearization}, plays a fundamental role in the proof of the Theorem and it is a generalization of a similar lemma proved in~\cite{stilwell2006sufficient}.
\begin{lemma}
\label{eq:lemma1}
 Let us consider a $T-$periodic matrix-valued function ${\bf E}(t)$ and let ${\langle {\bf E}\rangle}$ be its average, i.e.,
\begin{equation}
    \label{eq:Eaverage}
    \langle {\bf E}\rangle=\frac{1}{T}\int_{t}^{t+T}{\bf E}(t)\;dt
\end{equation}
for all $t>t_0$.

Suppose there exist a $C^{1}$ matrix-valued function ${\bf B}(t)$, positive scalars $M_{{\bf B}}$, $L_{{\bf B}}$ such as $\|{\bf B}(t)\|\le M_{{\bf B}}$, $\|d{{\bf B}}/dt(t)\|\le L_{{\bf B}}$ and
\begin{equation}
    \label{eq:averagesys}
    \dot{{\bm x}}(t) = \left[{\bf A}(t) +\langle {\bf E}\rangle{\bf B}(t)\right]{\bm x}(t),\quad {\bm x}(t_0)={\bm x}_0,\; t\ge t_0\, ,
\end{equation}
is uniformly exponentially stable. Then there exist $\varepsilon^*>0$ such as for all fixed $\varepsilon\in(0,\varepsilon^*)$
\begin{equation}
    \label{eq:tempsys}
    \dot{{\bm z}}(t) = \left[{\bf A}(t) +{\bf E}(t/\varepsilon){\bf B}(t)\right]{\bm z}(t),\quad {\bm z}(t_0)={\bm z}_0,\; t\ge t_0\, ,
\end{equation}
is uniformly exponentially stable.   
\end{lemma}

Let us observe that if ${\bf A}(t)$ and ${\bf B}(t)$ are constant matrices, then the uniform exponential stability Eq.~\eqref{eq:averagesys} can be determined by leveraging on the properties of its spectrum. Moreover, the existence of positive constants $M_{\bf B}$ and $L_{\bf B}$ is guaranteed whenever ${\bf B}(t)$ is periodic and continuously differentiable.
 
Assume to deal with a balanced $T$-periodic  directed $d$-uniform hypergraph, consequently also its averaged version~\eqref{eq:averadgetensor} will be balanced. Let us denote the averaged generalized Laplace matrix 
\begin{equation}
\label{eq:averagelaplace}
\left\langle \mathbf{L}^{(d)} \right\rangle
=
\frac{1}{T}
\int_{t}^{t+T}
\mathbf{L}^{(d)}(\tau)\,d\tau\, ,
\end{equation}
then we can state the follwong
\begin{theorem}
\label{thm:main}
Let ${\bm x}_s(t)$ denotes the synchronous solution to Eq.~\eqref{eq:general_problem_model_average}, let the linearized system about the latter solution be given by
\begin{equation}
\label{eq:general_linear_compact_average}
\frac{d}{dt}\delta\tilde{\mathbf{X}}
=
\left[
\mathbf{I}_{N}\otimes
\mathbf{J}_{0}\bigl({\bm x}_{s}(t)\bigr)
+
\mu\,
\left\langle \mathbf{L}^{(d)} \right\rangle
\otimes
\mathbf{J}^{(d)}_{\mathbf{h}}\bigl({\bm x}_{s}(t)\bigr)
\right]
\delta\tilde{\mathbf{X}}\, ,
\end{equation}
and assume the latter to be uniformly exponentially stable. Assume further the coupling function $\mathbf{h}$ to be smooth enough such that
$\mathbf{J}^{(d)}_{\mathbf{h}}\bigl(\bm{x}_{s}(t)\bigr)$
is a smoothly time-varying matrix whose norm and rate of variation are
bounded for all $t$. 

Then, there exists $\varepsilon^{*}>0$ such that, for every fixed
$\varepsilon\in(0,\varepsilon^{*})$, the ${\bm x}_s(t)$ denotes the synchronous solution to
Eq.~\eqref{eq:general_problem_model} is uniformly exponentially stable.
\end{theorem}

\begin{proof}
Let us to denote by $\langle \mathbf{U} \rangle$ the Schur transformation of $\left\langle \mathbf{L}^{(d)} \right\rangle$. Then, there exists a unitary matrix $\mathbf{P}$ such that $\langle \mathbf{U} \rangle = \mathbf{P}^{-1}\left\langle \mathbf{L}^{(d)} \right\rangle\mathbf{P}$, with 
\begin{equation}
\label{eq:SchurTransf}
    \langle \mathbf{U} \rangle
    =
    \begin{pmatrix}
        0
        &
        \mathbf{0}_{1\times(N-1)}
        \\
        \mathbf{0}_{(N-1)\times 1}
        &
        \langle \mathbf{U}_s \rangle\, 
    \end{pmatrix}\,,
\end{equation}
where $\langle \mathbf{U}_s \rangle\in\mathbb{C}^{(N-1)\times(N-1)}$ is upper triangular. Let us define a transformed perturbation vector by using $\mathbf{P}$, namely $\delta \tilde{\mathbf{X}}=(\mathbf{P}^{-1}\otimes\mathbf{I}_n)\delta\tilde{\mathbf{Y}}$; then Eq.~\eqref{eq:general_linear_compact_average} can be split into a mode parallel to the synchronization manifold and the remaining ones transverse to it. More specifically, we have
\begin{equation}
\label{eq:SLLinearmodelcompactformhypo_trans_avera}
\frac{d}{dt}\;\delta\tilde{\mathbf{Y}}_{\parallel}
=\mathbf{J}_{0}(\bm{x}_{s}(t))
\delta\tilde{\mathbf{Y}}_{\parallel}\, ,
\end{equation}
and
\begin{equation}
\label{eq:transverse_average}
\frac{d}{dt}\;\delta\tilde{\mathbf{Y}}_{\bot}
=
\Bigg[
\mathbf{I}_{N-1}
\otimes
\mathbf{J}_{0}(\bm{x}_{s}(t))
+
\mu\,
\langle\mathbf{U}_s\rangle
\otimes
\mathbf{J}^{(d)}_{\mathbf{h}}(\bm{x}_{s}(t))
\Bigg]
\delta\tilde{\mathbf{Y}}_{\bot}\, .
\end{equation}

Under the assumption of the balanced directed higher-order structures, we can use the same unitary matrix $\mathbf{P}$ to transform the Laplace matrix, $\mathbf{L}^{(d)}(t)$, namely $\mathbf{U}(t) = \mathbf{P}^{-1}\mathbf{L}^{(d)}(t)\mathbf{P}$ where
\begin{equation}
\label{eq:SchurTransfb}
     \mathbf{U}(t) 
    =
    \begin{pmatrix}
        0
        &
        \mathbf{0}_{1\times(N-1)}
        \\
        \mathbf{0}_{(N-1)\times 1}
        &
        \mathbf{U}_s(t)
    \end{pmatrix}\, .
\end{equation}
By starting from Eq.~\eqref{eq:general_linear_compact}, we can obtain the parallel and transverse modes evolution
\begin{equation}
\label{eq:SLLinearmodelcompactformhypo_trans_temp}
\frac{d\,\delta\mathbf{Y}_{\parallel}}{dt}
=\mathbf{J}_{0}(\bm{x}_{s}(t))
\delta\mathbf{Y}_{\parallel}\, ,
\end{equation}
and
\begin{equation}
\label{eq:transverse_}
\frac{d\,\delta\mathbf{Y}_{\bot}}{dt}
=
\Bigg[
\mathbf{I}_{N-1}
\otimes
\mathbf{J}_{0}(\bm{x}_{s}(t))
+
\mu\,
\mathbf{U}_s(t/\varepsilon)
\otimes
\mathbf{J}^{(d)}_{\mathbf{h}}(\bm{x}_{s}(t))
\Bigg]
\delta\mathbf{Y}_{\bot}\, .
\end{equation}

Let us observe that $\mathbf{U}_s(t)$ is not necessarily triangular, but its time-averaged does. Indeed,
\begin{align} \notag
\langle \mathbf{U} \rangle
&= \mathbf{P}^{-1}\langle \mathbf{L}^{(d)} \rangle\mathbf{P} = \frac{1}{T}\int_t^{t+T}
\mathbf{P}^{-1}\mathbf{L}^{(d)}(\tau)\mathbf{P}\,d\tau.
\end{align}
Then,
\begin{equation}
\langle \mathbf{U}_s \rangle
=
\frac{1}{T}\int_t^{t+T}
\mathbf{U}_s(\tau)\,d\tau \,.
\end{equation}
By hypothesis, the trivial solution of Eq.~\eqref{eq:transverse_average} is uniformly exponentially stable. By a direct application of Lemma~\ref{eq:lemma1}, it follows that the trivial solution of Eq.~\eqref{eq:transverse_} is also uniformly exponentially stable for every fixed $0<\varepsilon<\varepsilon^\ast$. Equivalently, the synchronization manifold $
\bm{x}_1(t)=\bm{x}_2(t)=\cdots=\bm{x}_N(t)=\bm{x}_s(t)
$ of Eq.~\eqref{eq:general_problem_model} is uniformly exponentially stable.
\end{proof}

\section{Application to Stuart-Landau oscillators}
\label{sec:2}

The aim of this Section is to introduce a model of Stuart--Landau oscillators coupled via a higher-order network and with a nonlinear coupling function $\mathbf h$. We will first consider the simpler case where $\mathbf h$ is resonant and thus $\mathbf{J}^{(d)}_{\mathbf h}(t)$ is constant; then we will relax the assumption by considering non-resonant coupling, so to obtain a $T-$periodic Jacobian, $\mathbf{J}^{(d)}_{\mathbf h}(t)$.

\subsection{The Model}

Let us consider a population of $N$ identical
Stuart--Landau (SL) oscillators interacting through a time-varying directed
$d$-uniform hypergraph. The state of oscillator $\ell$ is described by
a complex amplitude $w_\ell\in\mathbb{C}$, with
$\ell=1,\ldots,N$. Its dynamics is given by
\begin{align}
\label{eq:SLmodelDirecHypGraph}
\frac{dw_{\ell}}{dt}
&=
\mathbf{f}(w_\ell)
+\\ \notag
&+\mu
\sum_{j_1,\ldots,j_d=1}^{N}
\mathbf{g}^{(d)}
\left(
w_\ell,w_{j_1},\ldots,w_{j_d}
\right)
\mathbf{A}^{(d)}_{\ell j_1\cdots j_d}(t/\varepsilon)\, ,
\end{align}
where $\mathbf{f}(w_\ell)
=
\sigma w_\ell
-
\beta |w_\ell|^2w_\ell$ is the Stuart--Landau vector field,
$\sigma=\sigma_{\Re}+i\sigma_{\Im}$ and
$\beta=\beta_{\Re}+i\beta_{\Im}$ are complex parameters, and
$\mu=\mu_{\Re}+i\mu_{\Im}$ denotes the complex coupling strength.
The function
$\mathbf{g}^{(d)}:\mathbb{C}^{d+1}\rightarrow\mathbb{C}$
describes the nonlinear multivariate interaction associated with a
directed $d$-hyperedge. In the \textit{supercritical regime} corresponding to  $\sigma_{\Re}>0$ and $\beta_{\Re}>0$, each isolated SL oscillator admits a stable limit-cycle (LC) solution~\cite{nakao2014complex,kuramoto1984chemical} of the form $w_{\mathrm{LC}}(t)=\left|w_{\mathrm{LC}}\right| \exp(i\boldsymbol{\omega} t)$, with amplitude $\left|w_{\mathrm{LC}}\right|=\displaystyle\sqrt{\sigma_{\Re}/\beta_{\Re}}$, and frequency $\boldsymbol{\omega}=\sigma_{\Im}-\beta_{\Im}\displaystyle \sigma_{\Re}/\beta_{\Re}$. Global synchronization of the coupled SL oscillators requires
$
w_\ell(t)=w_{\mathrm{LC}}(t),
\text{ with } \ell=1,\ldots,N,
$
to be a linearly stable solution of Eq.~\eqref{eq:SLmodelDirecHypGraph} for a nonzero coupling strength. This solution defines the synchronization manifold. To ensure that such a solution exists, we suppose that $\mathbf{g}^{(d)}:\mathbb{C}^{d+1} \to \mathbb{C}$ is a diffusive-like complex-valued function, and we assume $\mathbf{h}^{(d)}$ to be given by the complex monomial coupling function
\begin{equation}
\label{eq.comxmonofunct}
\mathbf{h}^{(d)}\left(w_{j_1},\ldots,w_{j_d}\right)
=
\prod_{s=1}^{d}
w_{j_s}^{a_s}\bar{w}_{j_s}^{b_s}\, ,
\end{equation}
with $a_s$ and $b_s$ positive integers.  

To establish the stability of the synchronization manifold, we perform a linear stability analysis in the vicinity of the LC solution. We therefore consider a heterogeneous perturbation around this solution $w_{\textrm{LC}}(t)$
\begin{equation}
\label{eq:perturb}
w_\ell(t) = w_{\textrm{LC}}(t)(1 + \rho_\ell(t))e^{i\theta_\ell(t)}, \quad \ell=1,\dots, N\, ,
\end{equation}
where $\rho_\ell$ and $\theta_\ell$ are small perturbations respectively of the amplitude and phases whose time evolution can be obtained by substituting the above expression into Eq.~\eqref{eq:SLmodelDirecHypGraph} and retrieving only first order terms
\begin{equation}
\left\{
\begin{aligned}
\frac{d \rho_\ell}{dt}
&= -2\sigma_{\Re}\,\rho_\ell
+ \lvert w_{\text{LC}} \rvert^{-1}
\Re\!\left[ e^{-i\boldsymbol{\omega} t}\,\mathscr{L}_\ell \right], \\[1.2ex]
\frac{d \theta_\ell}{dt}
&= -2\beta_{\Im}\dfrac{\sigma_{\Re}}{\beta_{\Re}}\,\rho_\ell
+ \lvert w_{\text{LC}} \rvert^{-1}
\Im\!\left[ e^{-i\boldsymbol{\omega} t}\,\mathscr{L}_\ell \right],
\end{aligned}
\right.
\end{equation}
with $\mathscr{L}_\ell$ identifying the linearization of the coupling term in Eq.~\eqref{eq:SLmodelDirecHypGraph}. A straightforward derivation allows us to obtain the linear form of the coupling function $\mathbf{h}^{(d)}$ (see Appendix~\ref{App:linearization}), namely 
\begin{equation}
\label{eq:Linearizedcoufunct}
\begin{aligned}
\mathbf{h}^{(d)}\left(w_{j_1},\ldots,w_{j_d}\right)
&\sim
\left|w_{\mathrm{LC}}\right|^{\boldsymbol{\alpha}}
e^{i\boldsymbol{\gamma}\boldsymbol{\omega} t}
\Biggl[
1+\sum_{s=1}^{d}(a_s+b_s)\rho_{j_s}
\\
&
+i\sum_{s=1}^{d}(a_s-b_s)\theta_{j_s}
\Biggr],
\end{aligned}
\end{equation}
where we defined $\boldsymbol{\alpha} = \sum_s (a_s + b_s)$ and $\boldsymbol{\gamma} = \sum_s (a_s - b_s)$. Thus, in general, the model can exhibit two frequencies: one specific to the Stuart-Landau oscillator, and the second induced by the coupling term namely $\boldsymbol{\gamma} \boldsymbol{\omega}$. After some straightforward calculations, we can deduce  
\\
{\small
\begin{equation}
\label{eq:Linearizedcouterm}
\begin{split}
|w_{\mathrm{LC}}|^{-1}e^{-i\boldsymbol{\omega} t}\,\mathscr{L}_\ell
&\sim
|w_{\mathrm{LC}}|^{\boldsymbol{\alpha}-1}\Bigl[
\boldsymbol{\alpha} T_1(t)
\sum_{j} \mathbf{L}^{(d)}_{\ell j}(t/\varepsilon)\,\rho_j
\\
&\hspace{-1.5em}
-\boldsymbol{\gamma} T_2(t)
\sum_{j} \mathbf{L}^{(d)}_{\ell j}(t/\varepsilon)\,\theta_j
+i\boldsymbol{\gamma} T_1(t)
\sum_{j} \mathbf{L}^{(d)}_{\ell j}(t/\varepsilon)\,\theta_j
\\
&\hspace{-1.5em}
+i\boldsymbol{\alpha} T_2(t)
\sum_{j} \mathbf{L}^{(d)}_{\ell j}(t/\varepsilon)\,\rho_j
\Bigr],
\end{split}
\end{equation}
}
where  
\begin{equation}
\label{eq:T1T2}
\begin{aligned}
T_1(t) &= \mu_{\Re}\cos[(\boldsymbol{\gamma}-1)\boldsymbol{\omega} t] - \mu_{\Im}\sin[(\boldsymbol{\gamma}-1)\boldsymbol{\omega} t],\\
T_2(t) &= \mu_{\Im}\cos[(\boldsymbol{\gamma}-1)\boldsymbol{\omega} t] + \mu_{\Re}\sin[(\boldsymbol{\gamma}-1)\boldsymbol{\omega} t]. 
 \end{aligned}
\end{equation}
Eventually, we get the following linear expression
\begin{widetext}
\begin{equation} \label{eq:SLLinearmodelDirectHyperGraph2}
    \left\{
    \begin{aligned}
        \frac{d{\rho}_\ell}{dt} &= -2\sigma_{\Re}\rho_\ell 
        + |w_{\text{LC}}|^{\boldsymbol{\alpha} - 1} \left[ \boldsymbol{\alpha} T_1(t)\sum_j \mathbf{L}^{(d)}_{\ell j}(t/\varepsilon)
        \rho_j -  \boldsymbol{\gamma} T_2(t)\sum_j \mathbf{L}^{(d)}_{\ell j}(t/\varepsilon) 
        \theta_j \right], \\
        \frac{d{\theta}_\ell}{dt} &= -2\beta_{\Im}\dfrac{\sigma_{\Re}}{\beta_{\Re}}\rho_\ell 
         + |w_{\text{LC}}|^{\boldsymbol{\alpha} - 1} \left[\boldsymbol{\gamma} T_1(t)\sum_j \mathbf{L}^{(d)}_{\ell j}(t/\varepsilon) 
        \theta_j + \boldsymbol{\alpha} T_2(t)\sum_j \mathbf{L}^{(d)}_{\ell j}(t/\varepsilon) 
        \rho_j\right]\, .
    \end{aligned}
    \right. 
\end{equation}
\end{widetext}

By using a stack state vector $\mathcal{E} = (\rho_1, \theta_1, \dots , \rho_N, \theta_N)$, we can rewrite the latter equation as follows 
\begin{equation}
\label{eq:SLLinearmodelcompactform}
\begin{aligned}
\frac{d \mathcal{E}}{dt}
&=
\Bigg[ \mathbf{I}_N \otimes
\begin{pmatrix}
    -2\sigma_{\Re} & 0 \\
    -2\beta_{\Im}\dfrac{\sigma_{\Re}}{\beta_{\Re}} & 0
\end{pmatrix}+\\
&+ |w_{\text{LC}}|^{\boldsymbol{\alpha}-1}\mathbf{L}^{(d)}(t/\varepsilon)\otimes
\begin{pmatrix}
    \boldsymbol{\alpha} T_{1}(t) & -\boldsymbol{\gamma} T_{2}(t) \\
   \boldsymbol{\alpha} T_{2}(t) & \boldsymbol{\gamma} T_{1}(t)
\end{pmatrix}
\Bigg]
\mathcal{E}
\, .
\end{aligned}
\end{equation}

Note that, the Stuart--Landau dynamics is invariant under a global phase rotation,
$w_\ell \mapsto w_\ell e^{i\phi}$, reflecting the rotational symmetry of
its limit cycle. To preserve this symmetry in the coupled system, we can
require the nonlinear coupling to transform consistently with the
Stuart--Landau dynamics. For the monomial coupling considered here,
this is achieved by imposing the \textit{resonance condition}
\begin{equation}
\label{eq:gamma1}
\boldsymbol{\gamma} = \sum_s (a_s - b_s)= 1\, .
\end{equation}
The latter condition implies that the functions $T_1(t)$ and $T_2(t)$ given in Eq.~\eqref{eq:T1T2} are in reality constant ones, and thus the temporal dependence of the linearized coupling is entirely carried by the Laplace matrix. Indeed inserting the resonance condition of Eq.~\eqref{eq:gamma1} into Eq.~\eqref{eq:T1T2}, we get $T_1(t)=\mu_{\Re}$ and $T_2(t)=\mu_{\Im}$. Then,
\begin{equation}
\label{eq:SLLinearmodelcompactformhypo}
\begin{aligned}
\frac{d\mathcal{E}}{dt}
&=
\Bigg[
\mathbf{I}_N \otimes
\begin{pmatrix}
-2\sigma_{\Re} & 0 \\
-2\beta_{\Im}\dfrac{\sigma_{\Re}}{\beta_{\Re}} & 0
\end{pmatrix}+
\\
&+ |w_{\mathrm{LC}}|^{\boldsymbol{\alpha}-1}
\mathbf{L}^{(d)}(t/\varepsilon)\otimes
\begin{pmatrix}
\boldsymbol{\alpha} \mu_{\Re} & -\mu_{\Im} \\
\boldsymbol{\alpha} \mu_{\Im} & \mu_{\Re}
\end{pmatrix}
\Bigg]
\mathcal{E}\, .
\end{aligned}
\end{equation}
Consequently, when the switching is sufficiently fast, i.e., for sufficiently small values of \(\varepsilon\), the exponential stability of the averaged counterpart of Eq.~\eqref{eq:SLLinearmodelcompactformhypo} can be established analytically, since both the local and coupling Jacobian matrices are constant ones.

In the non-resonant case
\begin{equation}
    \label{eq:nonresonant}
   \boldsymbol{\gamma} = \sum_s (a_s - b_s)\ne 1\, ,
\end{equation}
the averaged Jacobian of Eq.~\eqref{eq:SLLinearmodelcompactform} is given by
\begin{equation}
 \begin{aligned}
       \langle{\bf J}\rangle&=\mathbf{I}_N \otimes
\begin{pmatrix}
    -2\sigma_{\Re} & 0 \\
    -2\beta_{\Im}\dfrac{\sigma_{\Re}}{\beta_{\Re}} & 0
\end{pmatrix}+\\
&+ |w_{\text{LC}}|^{\boldsymbol{\alpha}-1}\langle\mathbf{L}^{(d)}\rangle\otimes
\begin{pmatrix}
    \boldsymbol{\alpha} T_{1}(t) & -\boldsymbol{\gamma} T_{2}(t) \\
   \boldsymbol{\alpha} T_{2}(t) & \boldsymbol{\gamma} T_{1}(t)
\end{pmatrix}\, ,
 \end{aligned}
\end{equation}
and it results to be $T'=2\pi /|\boldsymbol{\gamma}-1|$-periodic. The stability of  trivial solution can be deduced by applying the Floquet theory~\cite{rspa.2025.0907}.

To illustrate the theoretical framework developed in the previous sections, we consider the time evolution of the underlying structure to be obtained by switching between a finite set of directed hyperrings (see Fig.~\ref{fig1} for illustration); namely higher-order structures where nodes can be arranged in a ring and therefore characterized by circulant generalized Laplacian matrices. Since all circulant matrices are diagonalizable by the Fourier basis, temporal hypergraphs switching between different hyperring configurations share a common eigenbasis. Moreover, the Laplacian eigenvalues can be explicitly expressed in terms of the structural parameters, allowing us to identify parameter regimes in which the instantaneous hypergraphs are individually unstable while their averaged counterpart is stable.
\begin{figure}[!ht]
    \centering
    \includegraphics[width=1\linewidth]{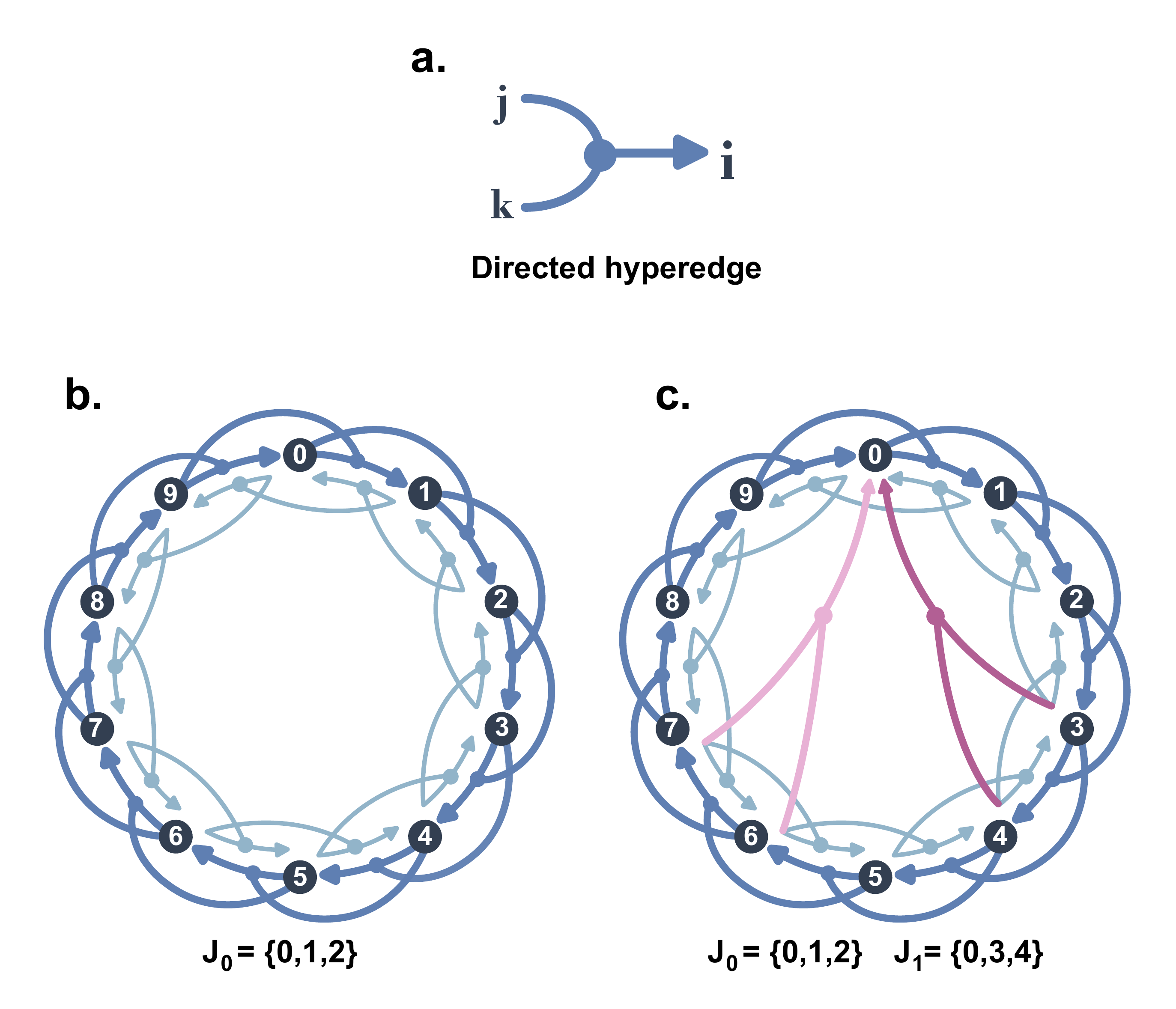}
    \caption{\textbf{Balanced directed $d$-uniform hyperring}. Panel \textbf{a.} Illustration of a directed $2-$hyperedge, nodes $j$ and $k$ form the tail and act on $i$ in the head of the hyperedge. Panel \textbf{b.} Construction of a directed $2$-uniform hyperring starting from the selected nodes $\{0,1,2\}$. We construct two familly of hyperedges, one consisting of dark blue hyperedges oriented in the $``+"$ direction, and the other consisting of light blue hyperedges oriented in the $``-"$ direction.  Panel \textbf{c.} Display the first two hyperedges that are formed when considering the set  $\{0,3,4\}$.
}
    \label{fig1}
\end{figure}

\subsection{Spectral properties of balanced directed $d$-uniform hyperring}\label{subsec:spectBalhyper}

The goal of this section is to provide an algorithm to build a $d$-uniform balanced $1$-directed hypergraph whose generalized lapalcian matrix is balanced, possibly weighted, with $N$ nodes, and then we provide an explicit expression for its spectrum. Roughly speaking the idea is to first fix an head node, say $0$, $d$ tails nodes, $j_1^{(0)}, \dots, j_d^{(0)}$ and a weight $p_0>0$, in this way we define an hyperedge encoded by the adjacency tensor $\mathbf{A}^{(d)}_{0 j_1^{(0)}\cdots j_d^{(0)}}
=p_0$. Then we ``slide'' the hyperedge by creating $N-1$ sets of $(d+1)$ nodes obtained by adding to all the above labels, $J_0^+=\{0, j_1^{(0)},\cdots, j_d^{(0)}\}$, the integers $\ell=1,\dots,N-1$, i.e., $J_\ell^{+}=\{\ell, j_1^{(0)}+\ell,\cdots, j_d^{(0)}+\ell\}$. The latter hyperedges can be encoded into the adjacency tensor
\begin{equation}
\label{eq:adjell}
\mathbf{A}^{(d)}_{\ell j_1^{(0)}+\ell\cdots j_d^{(0)}+\ell}
=p_0,\quad \ell=1,\dots,N-1\, .
\end{equation}
Let us remember that node indexes should be considered modulo $N$ and that we should avoid repeated indexes in the hyperedges, sort of self-loops, and moreover the hyperedges must be distinct, i.e., avoid sort of multi-link. The latter requirements induces some constraints on the choice of the initial indexes $j_1^{(0)}\cdots j_d^{(0)}$ and the size of the hypergraph $N$. 
By using the above construction we obtain from Eq.~\eqref{eq:genindeg} the generalized in-degree of nodes $0$
\begin{eqnarray*}
 k_0^{\mathrm{in}}&=&\frac{1}{d!}
\sum_{j_1,\ldots,j_d=1}^{N-1}
\mathbf{A}^{(d)}_{0 j_1\cdots j_d}=
p_0\sum_{j_1<\ldots<j_d}^{N-1}\delta_{j_1j_1^{(0)}}\dots \delta_{j_dj_d^{(0)}}\notag\\
&=&p_0\, ,
\end{eqnarray*}
namely the entries of the adjacency tensor are zero but the one corresponding to the set of indexes $j_1^{(0)}\cdots j_d^{(0)}$, whose value is $p_0$. By using the cyclic construction, the same holds true for all nodes $\ell=1,\dots,N-1$. Similarly we can compute the total weight of directed hyperedges containing node $0$ in the head and node $j$ in the tail from Eq.~\eqref{eqkinellj}
\begin{align}
k_{0j}
&=
\frac{1}{(d-1)!}
\sum_{j_2,\ldots,j_d=1}^{N-1}
\mathbf{A}^{(d)}_{0j j_2\cdots j_d}
\\ \notag
&=
p_0
\sum_{j_2<\cdots<j_d}^{N-1}
\delta_{j_2 j_2^{(0)}}\cdots
\delta_{j_d j_d^{(0)}}
=
p_0.
\end{align}
Hence $k_{0 j}=p_0$ for $j\in\{j_1^{(0)},\dots,j_d^{(0)}\}$. In conclusion the generalized Laplacian matrix defined in Eq.~\eqref{eq:Lapd} is given by
\begin{equation}
    \mathbf{L}^{(d)}_{0 j}
=
\begin{cases}
-p_0 d!
& j=0,\\[2mm]
p_0(d-1)!
& j\in\{j_1^{(0)},\dots,j_d^{(0)}\}\, .
\end{cases}
\end{equation}
The other entries can be computed similarly by using the cyclic construction defined above. 

We can also compute the total weight of hyperedges that have node $0$ in tail from Eq.~\eqref{eq:kout} to obtain
\begin{equation*}
 k_0^{\mathrm{out}}
= dp_0\, ,
\end{equation*}
and, again, by cyclic construction $k_\ell^{\mathrm{out}}
= dp_0$ for all $\ell=1,\dots,N-1$. We have thus proved that
\begin{equation*}
    k_\ell^{\mathrm{out}}=dp_0=dk_\ell^{\mathrm{in}},\quad \ell=0,\dots,N-1\, ,
\end{equation*}
namely the hypergraph is balanced and circulant by construction. Importantly, this construction can be viewed as a natural generalization of directed ring networks to the setting of directed hypergraphs.

Let us observe that a similar construction can be obtained by ``sliding'' the basic hyperedge, $J_0^-=\{0, j_1^{(0)}-N,\cdots, j_d^{(0)}-N\}$, on the other direction, i.e., by creating $N-1$ sets of $(d+1)$ nodes obtained by removing all the above labels from the integers $\ell=1,\dots,N-1$. Hence, we obtain a new family of directed hyperedges $J_\ell^{-}=\{\ell, \ell-j_1^{(0)}+N,\cdots, \ell-j_d^{(0)}+N\}$ with the indices taken modulo $N$, and with 
\begin{equation}
\label{eq:adjellm}
\mathbf{A}^{(d)}_{\ell (\ell-j_1^{(0)}+N)\cdots (\ell-j_d^{(0)}+N)}
=q_0,\quad \ell=1,\dots,N-1\, ,
\end{equation}
for some $q_0>0$. Generally speaking, we can eventually build a balanced directed $d$-uniform hyperring by combining different basic hyperedges, $J_r=\{0,j_1^{(r)},\dots, j_d^{(r)}\}$,  each one with its own weight and apply the shift as described above. The resulting adjacency tensor will the sum of the ones obtained from each single shift.

Suppose we generate the hypergraph starting from several distinct basic hyperedges $J_r$, where $r = 1, \dots, p$, and such that for each hyperedge we construct the families of hyperedges in both two possible directions. Thanks to the circular nature of the matrix $\mathbf{L}^{(d)}$, we can diagonalize the latter via the Fourier basis, whose $k$-th eigenvector, $k=1,\ldots,N$,  is given by
\begin{equation}
    \Psi^{(k)}=\frac{1}{\sqrt{N}}
    \begin{pmatrix}
        1\\
        e^{i\varphi_k}\\
        e^{2i\varphi_k}\\
        \vdots \\
        e^{(N-1)i\varphi_k}
    \end{pmatrix},
\end{equation}
where $\varphi_k=2\pi(k-1)/N$.
The mode $k=1$ is associated to the homogenous vector and corresponds to
$\Lambda_1=0$. For any node $j$, we have $\Psi^{(k)}_{\ell+j} = e^{ij\varphi_k}\Psi^{(k)}_{\ell}$ and $\Psi^{(k)}_{\ell-j} = e^{-i j\varphi_k}\Psi^{(k)}_{\ell}$.  Thus, we can straightforward get $\mathbf{L}^{(d)}\Psi^{(k)}= \Lambda_k\Psi^{(k)}$ with 
\begin{equation}
    \Lambda_k = (d-1)!\sum_{r=1}^p\Big[p_r\sum_{j\in\mathbf{J}_r}\big(e^{ij\varphi_k}-1\big) +q_r\sum_{j\in\mathbf{J}_r}\big(e^{-ij\varphi_k}-1\big) \Big]\, ,
\end{equation}
from which we can extract the real and imaginary parts:
\begin{equation}
\mathrm{Re}\Lambda_k
=
(d-1)!\sum_{r=1}^{p}
(p_r+q_r)
\sum_{j\in \mathbf{J}_r}
\left[
\cos(j\theta_k)-1
\right]\, ,
\end{equation}
and
\begin{equation}
\mathrm{Im}\Lambda_k
=
(d-1)!\sum_{r=1}^{p}
(p_r-q_r)
\sum_{j\in \mathbf{J}_r}
\sin(j\theta_k)\, .
\end{equation}

\subsection{Numerical results}

\begin{figure*}[!ht]
    \centering
    \includegraphics[width=1\linewidth]{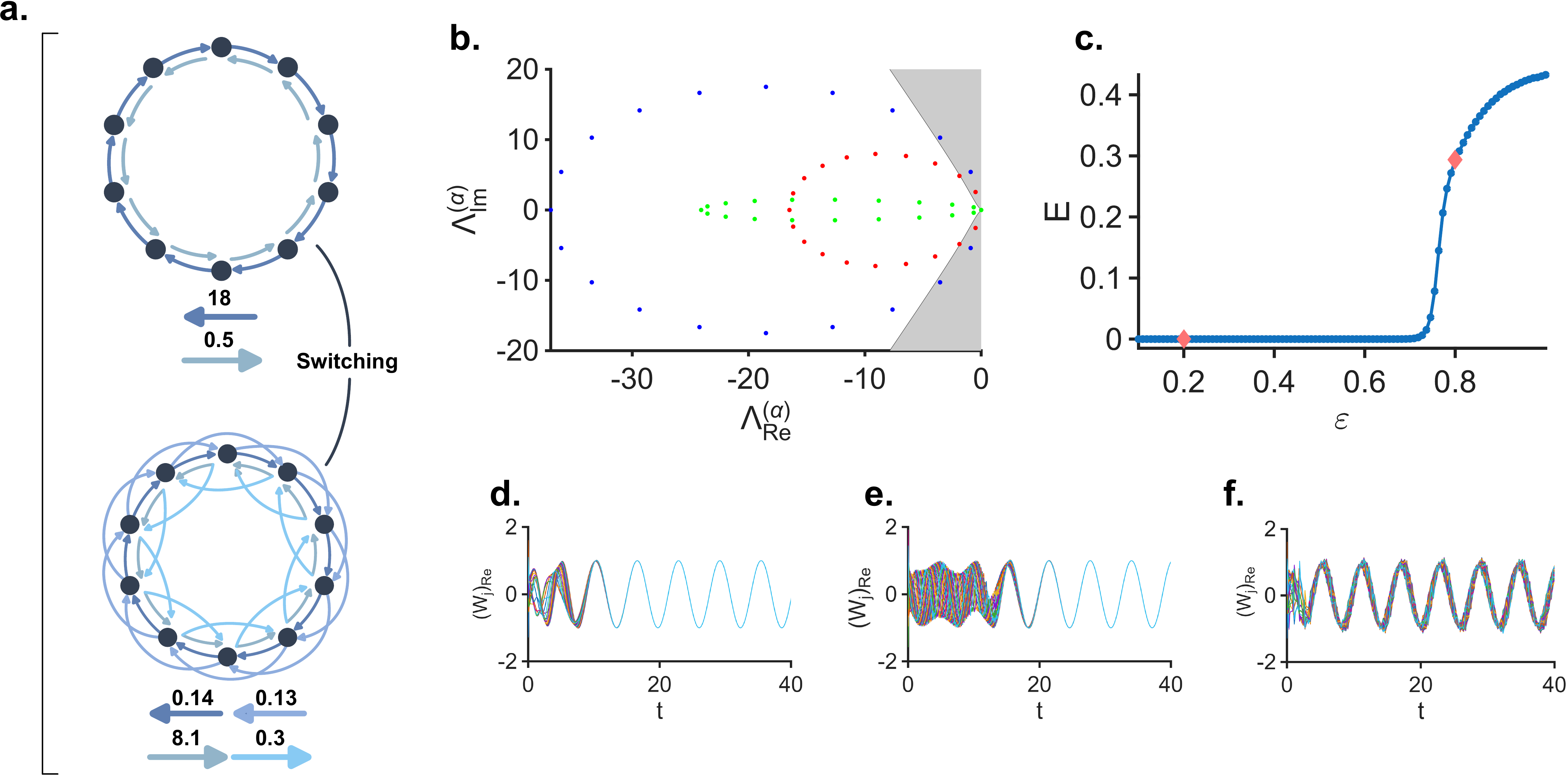}
    \caption{
\textbf{Synchronization in a time-varying balanced directed network.}
\textbf{a}, Schematic configurations with $N=10$ nodes.
The upper network has $J_1^+=\{0,1\}$, $p_1=18$, and $J_1^-=\{0,9\}$ with $q_1=0.5$.
The lower network has $J_1^+=\{0,1\}$, $J_2^+=\{0,2\}$,
$p_1=0.14$, $p_2=0.13$, and $J_1^-=\{0,9\}$, $J_2^-=\{0,8\}$ with $q_1=8.1$, and $q_2=0.3$.
Panels~\textbf{b}--\textbf{f} use the same construction with $N=20$.
\textbf{b}, Laplacian spectra of the first configuration
(blue), the second configuration (red), and their weighted
average (green), with $\zeta=0.37$.
The grey region denotes $\mathcal{M}(\Lambda)>0$;
the white region away from the boundary denotes
$\mathcal{M}(\Lambda)<0$.
\textbf{c}, Time-averaged synchronization error defined
in Eq.~\eqref{eq:syncerror}.
The highlighted points correspond to $\varepsilon=0.2$ and $0.8$.
\textbf{d}--\textbf{f}, Real parts of the oscillator states over time, respectively for the averaged network, the temporal network for
$\varepsilon=0.2$, and for $\varepsilon=0.8$.
Parameters:
$\mathbf{h}^{(1)}(w_j)=w_j^2\bar{w}_j$, $T=1$,
$\sigma=1+2i$, $\beta=1+i$, and $\mu=1-0.75i$.
}
    \label{fig2}
\end{figure*}

Let us first consider the resonant case, ${\bm \gamma}=1$. Let us consider SL oscillators coupled via an arbitrary static directed $d$-uniform hypergraph whose linearization is given by Eq.~\eqref{eq:SLLinearmodelcompactformhypo}, with a constant generalized Laplace matrix $\mathbf{L}^{(d)}$. By assuming the latter to be diagonalizable, we project the amplitude and phase perturbations onto its eigenbasis which leads 
$
\rho_j
=
\sum_{k=1}^{N}
\hat{\rho}_{k}\Psi_j^{(k)},
\text{ and }
\theta_j
=
\sum_{k=1}^{N}
\hat{\theta}_{k}\Psi_j^{(k)}.
$
Further details of this projection are provided in Appendix~\ref{App}. By substituting these expansions into Eq.~\eqref{eq:SLLinearmodelcompactformhypo} yields an independent two-dimensional system for each Laplacian mode,
\begin{equation}
\label{eq:SLLinearmodelGraphTimDep1}
\frac{d}{dt}
\begin{pmatrix}
\hat{\rho}_{k}\\
\hat{\theta}_{k}
\end{pmatrix}
=
\left(
\mathbf{J}_0
+
\Lambda_k\mathbf{J}_{\mathbf h}
\right)
\begin{pmatrix}
\hat{\rho}_{k}\\
\hat{\theta}_{k}
\end{pmatrix}
=:
\mathbf{J}_{k}
\begin{pmatrix}
\hat{\rho}_{k}\\
\hat{\theta}_{k}
\end{pmatrix},
\end{equation}
where
\begin{equation}
\label{eq:J0haut}
\mathbf J_0
=
\begin{pmatrix}
-2\sigma_{\Re} & 0\\[1mm]
-2\beta_{\Im}\dfrac{\sigma_{\Re}}{\beta_{\Re}} & 0
\end{pmatrix},
\text{ }
\mathbf J_{\mathbf h}
=
|w_\mathrm{LC}|^{\boldsymbol{\alpha}-1}
\begin{pmatrix}
\boldsymbol{\alpha}\mu_{\Re} & -\mu_{\Im}\\
\boldsymbol{\alpha}\mu_{\Im} & \mu_{\Re}
\end{pmatrix}.
\end{equation}

The Master Stability Function is obtained by replacing the discrete Laplacian eigenvalue $\Lambda_k$ in Eq.~\eqref{eq:SLLinearmodelGraphTimDep1} with a continuous complex parameter
$
\Lambda=\Lambda_{\Re}+i\Lambda_{\Im},
\text{ where }
\Lambda_{\Re}\leq 0,
\text{ and }
\Lambda_{\Im}\in\mathbb{R},
$
and by considering the one-parameter family of linear systems
\begin{equation}
\frac{d}{dt}
\begin{pmatrix}
\hat{\rho}\\
\hat{\theta}
\end{pmatrix}
=
\mathbf{J}_\Lambda
\begin{pmatrix}
\hat{\rho}\\
\hat{\theta}
\end{pmatrix},
\end{equation}
with $\mathbf{J}_\Lambda
=
\mathbf{J}_0+\Lambda\mathbf{J}_{\mathbf h}.$
The MSF is thus defined as the largest real part of the eigenvalues of $\mathbf{J}_\Lambda$,
\begin{equation}
\label{eq:MSF}
\mathcal{M}(\Lambda)
=
\max_{\nu\in\mathrm{spec}(\mathbf{J}_\Lambda)}
\Re(\nu).
\end{equation}
where $\mathrm{spec}(\mathbf{J}_\Lambda)$ denotes the spectrum of $\mathbf{J}_\Lambda$. The sign of MSF partitions the complex $\Lambda$-plane into stable and unstable regions. The synchronous solution is linearly stable if
$
\mathcal{M}\!\left(\Lambda_k\right)<0
\text{ for all transverse modes }
k=2,\ldots,N.
$
Conversely, if at least one nonzero Laplacian eigenvalue satisfies
$
\mathcal{M}\!\left(\Lambda_k\right)>0,
$
the corresponding transverse perturbation grows exponentially and the synchronized state is linearly unstable. The zero eigenvalue $\Lambda^{(1)}=0$ is associated with perturbations tangent to the synchronization manifold.

\begin{figure*}[!ht]
    \centering
    \includegraphics[width=1\linewidth]{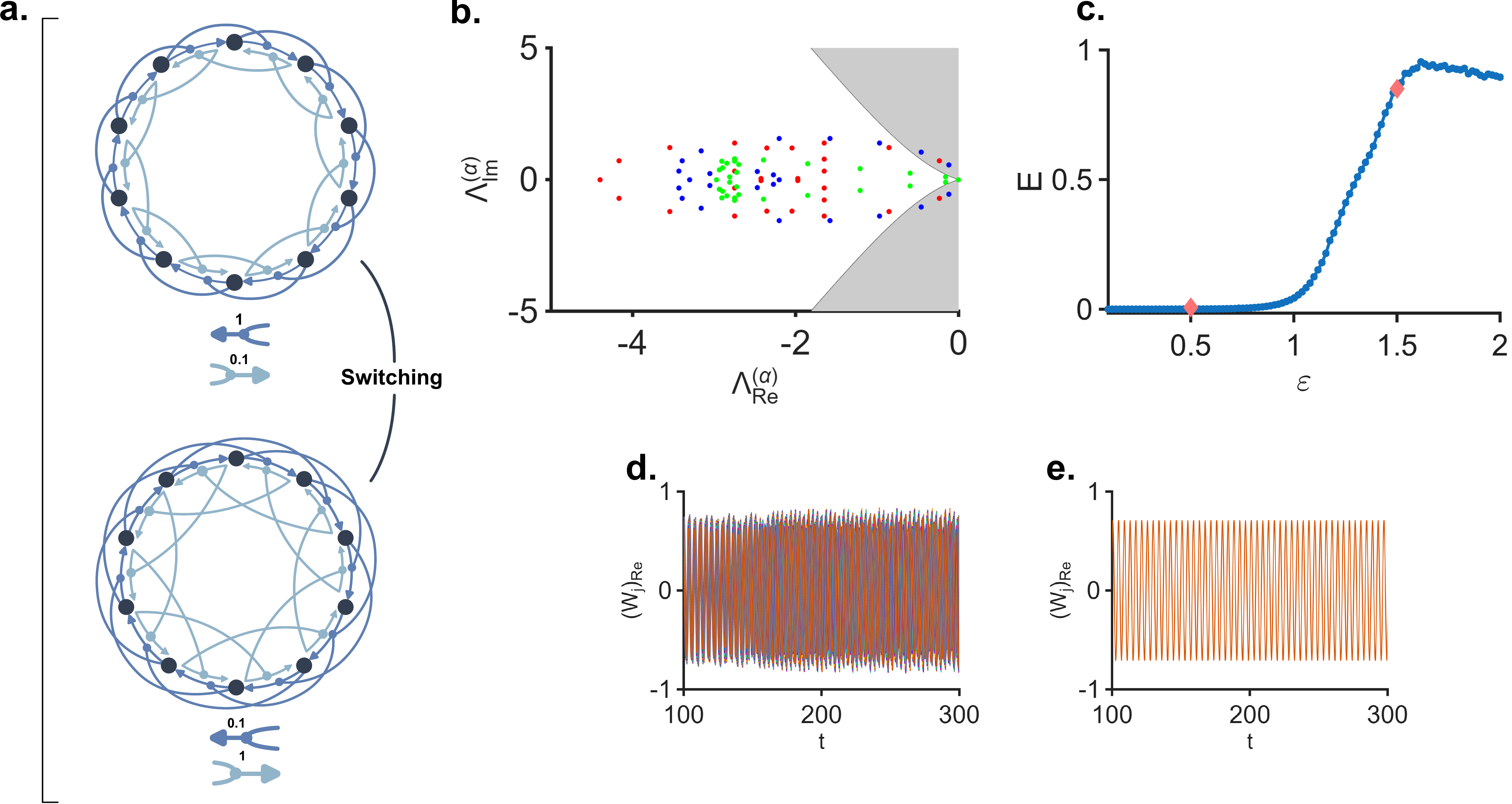}
    \caption{
\textbf{Synchronization in a time-varying balanced directed hyperring.}
\textbf{a}, Schematic configurations with $N=10$ nodes
and interactions of order $d=2$.
The upper hyperring has $J_1^+=\{0,1,2\}$ with
$p_1=1$ and $J_1^-=\{0,9,8\}$ with $q_1=0.1$.
The lower hyperring has  $J_1^+=\{0,1,3\}$
with $p_1=0.1$ and $J_1^-=\{0,9,7\}$ with $q_1=1$.
Panels~\textbf{b}--\textbf{e} use the same construction with $N=30$.
\textbf{b}, Laplacian spectra of the first configuration
(blue), the second configuration (red), and their weighted
average (green), with $\zeta=0.65$.
The grey region denotes $\mathcal{M}(\Lambda)>0$;
the white region away from the boundary denotes
$\mathcal{M}(\Lambda)<0$.
\textbf{c}, Time-averaged synchronization error defined
in Eq.~\eqref{eq:syncerror}.
The highlighted points correspond to $\varepsilon=0.7$ and $1.5$.
\textbf{d}, \textbf{e}, Real parts of the oscillator states
at $\varepsilon=0.7$ and $\varepsilon=1.5$, respectively.
Parameters:
$\mathbf{h}^{(2)}(w_{j_1},w_{j_2})
=w_{j_1}^{2}\bar{w}_{j_1}w_{j_2}\bar{w}_{j_2}$,
$T=2$, $\sigma=0.5+2i$, $\beta=1+i$, and $\mu=1-i$.
}
     \label{fig3}
\end{figure*}

Fig.~\ref{fig2} illustrates the fast-switching mechanism by using
two balanced directed circulant networks, corresponding to the
pairwise case of the construction introduced in
Subsection~\ref{subsec:spectBalhyper}.
Fig.~\ref{fig2}-\textbf{a} shows the two configurations with $N=10$ nodes
for visual clarity, whereas the numerical results in
Fig.~\ref{fig2}-(\textbf{b}--\textbf{f}) use $N=20$.
The first network is generated from $J_1^+=\{0,1\}$ and $J_1^-=\{0,9\}$, with
$p_1=18$ and $q_1=0.5$.
The second includes first- and second-neighbour interactions,
with $J_1^+=\{0,1\}$, $J_2^+=\{0,2\}$, and weights
$p_1=0.14$, $p_2=0.13$, and $J_1^-=\{0,9\}$, $J_2^-=\{0,8\}$ with $q_1=8.1$, and $q_2=0.3$.
The oscillator parameters are $\sigma=1+2i$, $\beta=1+i$,
and $\mu=1-0.75i$, and the nonlinear coupling
$\mathbf{h}^{(1)}(w_j)=w_j^2\bar{w}_j$ satisfies the resonance condition. Fig.~\ref{fig2}-\textbf{b} compares the spectra of the two instantaneous 
Laplacians with the one of their weighted average,
$\mathbf{L}^{(d)}_{\mathrm{avg}}
=\zeta\mathbf{L}^{(d)}_1+(1-\zeta)\mathbf{L}^{(d)}_2$,
with $\zeta=0.37$.
Being the latter matrices constant, we can explicitly compute the MSF; the gray region corresponds to $\mathcal{M}(\Lambda)>0$,
whereas the white region, excluding its boundary, corresponds
to $\mathcal{M}(\Lambda)<0$.
Both instantaneous spectra, shown in blue and red, exhibit
eigenvalues in the instability region.
Thus, the synchronized limit cycle is transversely unstable
on either static network.
By contrast, all transverse eigenvalues of the averaged
Laplacian, shown in green, lie in the stable region.
To assess synchronization, we integrate the temporal nonlinear dynamics for each value of $\varepsilon$ from the same initial condition close to the synchronized limit cycle and measure the dispersion of the oscillator states around their instantaneous mean,
$
m(t;\varepsilon)
=\sum_{j=1}^{N}w_j(t;\varepsilon)/N.
$
Then, the synchronization error is defined as
\begin{equation}
\label{eq:syncerror}
E(\varepsilon)
=
\left\langle
\sqrt{
\frac{1}{N|w_{\mathrm{LC}}|^2}
\sum_{j=1}^{N}
\left|
w_j(t;\varepsilon)-m(t;\varepsilon)
\right|^2
}
\right\rangle_{\mathcal{T}},
\end{equation}
where $\mathcal{T}$ is a final observation window after the
discarded transient. For Fig.~\ref{fig2}, $T=1$, so the physical switching period is
$T_{\mathrm{sw}}=\varepsilon$.
Fig.~\ref{fig2}-\textbf{c} shows an error close to zero for small
$\varepsilon$, followed by a pronounced increase around
$\varepsilon\simeq0.7$--$0.8$.
The small error in the fast-switching regime is consistent
with the local stability predicted from the averaged system.
At larger switching timescales, the nonzero error reveals
differences between the oscillator states that persist
over the observation window.
Fig.~\ref{fig2}-(\textbf{d}--\textbf{f}) complement this measurement:
the averaged network and the temporal network at
$\varepsilon=0.2$ show convergence towards a common oscillation,
whereas differences remain visible at $\varepsilon=0.8$.

Fig.~\ref{fig3} extends this analysis to balanced directed
hyperrings with interactions of order $d=2$, i.e., involving one
head node and two tail nodes per hyperedge.
The schematic configurations in Fig.~\ref{fig3}-\textbf{a} has been obtained with $N=10$ nodes,
while the numerical calculations have been realized with $N=30$ nodes.
The first hyperring is generated by $J_1^+=\{0,1,2\}$ with
$p_1=1$ and $J_1^-=\{0,9,8\}$ with $q_1=0.1$, and the second by $J_1^+=\{0,1,3\}$
with $p_1=0.1$ and $J_1^-=\{0,9,7\}$ with $q_1=1$. The nonlinear coupling is
$
\mathbf{h}^{(2)}(w_{j_1},w_{j_2})
=w_{j_1}^{2}\bar{w}_{j_1}w_{j_2}\bar{w}_{j_2}.
$ 
Its total degree is $\boldsymbol{\alpha}=5$, and its phase exponent satisfies
$(2-1)+(1-1)=1$.
Consequently, the resonance condition is guaranteed.
The remaining parameters are $\sigma=0.5+2i$, $\beta=1+i$,
and $\mu=1-i$. Fig.~\ref{fig3}-\textbf{b} compares the spectra of the two
directed hyperrings where
each instantaneous spectrum contains eigenvalues in the
gray instability region, so the synchronized limit cycle
is unstable on either static hyperring.
However, both generalized Laplacians are circulant and therefore share
the same Fourier eigenvectors.
Hence, one can explicitly compute the eigenvalues of the averaged Laplacian, and for each mode $k$ we obtain $\Lambda_{\mathrm{avg}}^k=\zeta\,\Lambda_{1}^k+(1-\zeta)\,\Lambda_{2}^k$ with $\zeta=0.65$.
For the chosen shifts and weights, these weighted
combinations place all transverse eigenvalues in the stable
region, as shown by the green points.
Thus, although each static configuration admits growing
transverse perturbations, their averaged dynamics suppresses
all transverse modes.
The fast-switching theorem then guarantees that this
transverse stability is inherited by the temporal system
for sufficiently small $\varepsilon$.

For the higher-order example we set the period to $T=2$, and hence
$T_{\mathrm{sw}}=2\varepsilon$.
Fig.~\ref{fig3}-\textbf{c} shows that the synchronization error remains
close to zero at sufficiently small switching timescales.
As $\varepsilon$ increases,
the error rises over a broader range than in Fig.~\ref{fig2},
reaching values close to $0.9$ at the upper end of the
displayed interval.
This nonzero error indicates persistent dispersion of the
complex states over the averaging window.
Conversely, the small-error regime provides numerical evidence
that sufficiently rapid switching can maintain synchronization
in the presence of directed higher-order interactions.
The real-part of the complex solutions shown in Fig.~\ref{fig3}-\textbf{d} and Fig.~\ref{fig3}-\textbf{e}
provide complementary illustrations, but synchronization is
assessed by  using Eq.~\eqref{eq:syncerror}, which accounts for both
components of the complex states.

Let us now consider the non-resonant case, namely,
$\boldsymbol{\gamma}\neq 1$. In this case,
Eq.~\eqref{eq:SLLinearmodelGraphTimDep1} reduces, for each Laplacian
eigenmode, to the time-periodic linear system
\begin{equation}
\label{eq:SLLinearmodelGraphFlo}
\frac{d}{dt}
\begin{pmatrix}
\hat{\rho}_{k}\\
\hat{\theta}_{k}
\end{pmatrix}
=
\left(
\mathbf{J}_{0}
+
\Lambda_{k}\mathbf{J}_{\mathbf h}(t)
\right)
\begin{pmatrix}
\hat{\rho}_{k}\\
\hat{\theta}_{k}
\end{pmatrix}
=:\
\mathbf{J}_{k}(t)
\begin{pmatrix}
\hat{\rho}_{k}\\
\hat{\theta}_{k}
\end{pmatrix},
\end{equation}
where
\begin{equation}
\label{eq:J0haut}
\mathbf{J}_{\mathbf h}(t)
=
|w_{\mathrm{LC}}|^{\boldsymbol{\alpha}-1}
\begin{pmatrix}
\boldsymbol{\alpha} T_{1}(t)
&
-\boldsymbol{\gamma} T_{2}(t)
\\
\boldsymbol{\alpha} T_{2}(t)
&
\boldsymbol{\gamma} T_{1}(t)
\end{pmatrix},
\end{equation}
with
$
\mathbf{J}_{\mathbf h}(t+T')
=
\mathbf{J}_{\mathbf h}(t).
$

As in the resonant case, we replace the discrete Laplacian eigenvalue
$\Lambda_k$ by a continuous complex parameter $\Lambda\in\mathbb{C}$ and
define $
\mathbf{J}_{\Lambda}(t)
=
\mathbf{J}_{0}
+
\Lambda\mathbf{J}_{\mathbf h}(t).
$
The corresponding stability region can then be determined numerically
by using Floquet theory. More precisely, for each value of $\Lambda$, we
integrate the matrix differential equation
\begin{equation}
\label{eq:systemfloquet}
\frac{d}{dt}\mathbf{M}(t)
=
\mathbf{J}_{\Lambda}(t)\mathbf{M}(t),
\end{equation}
over one period $T'$, with initial condition
$
\mathbf{M}(t_{0})=\mathbf{I}_{2}.
$
The matrix
$
\mathbf{M}(T')
$
is the corresponding monodromy matrix. Denoting its eigenvalues, i.e.,
the Floquet multipliers, by $\nu_j$, the associated Floquet exponents are
given by
\begin{equation}
\label{eq:expFloquet}
\vartheta_j
=
\frac{1}{T'}
\log(\nu_j).
\end{equation}
The Master Stability Function is therefore defined as the largest real
part among the Floquet exponents,
\begin{equation}
\mathcal{M}(\Lambda)
=
\max_j \Re(\vartheta_j)
.
\end{equation}
Hence, the instability
region in the complex $\Lambda$-plane is characterized by
$
\mathcal{M}(\Lambda)>0,
$
whereas $\mathcal{M}(\Lambda)<0$ corresponds to transverse stability.

We use again the temporal networks considered in Figs.~\ref{fig2} and~\ref{fig3}, but we now endow them with non-resonant coupling functions. For pairwise interactions, we choose
$
\mathbf{h}^{(1)}(w_j)=w_j^{5},
$ and with the parameters $\sigma=1$, $\beta=1+2i$, and $\mu=0.2(1+1.5i)$. These parameters are selected such that the averaged network is stable, whereas the two static topologies are individually unstable, as illustrated in Fig.~\ref{fig3}-\textbf{a}. The synchronization error further highlights the role of fast switching: it remains essentially zero for sufficiently small values of $\varepsilon$ and progressively increases as $\varepsilon$ becomes larger. This behavior is further illustrated in Figs.~\ref{fig3}-\textbf{c}--\textbf{d}, where the system is integrated for two different switching time scales. For $\varepsilon=0.05$, the sufficiently fast switching drives the system toward global synchronization, whereas for $\varepsilon=1.5$, the slower switching leads to desynchronization.
An analogous analysis is presented in Fig.~\ref{fig5} for higher-order interactions, where the coupling function is taken as
$
\mathbf{h}^{(2)}(w_{j_1},w_{j_2})
=
w_{j_1}^{2}\bar{w}_{j_1}
w_{j_2}$, with $T=2$, $\sigma=0.5+2i$, $\beta=1+i$, and $\mu=0.5\times(1-i)$. Once again, the numerical results demonstrate that sufficiently fast temporal switching can stabilize the synchronous state even when the individual frozen topologies do not support synchronization.

\begin{figure}[!hb]
    \centering
    \includegraphics[width=1\linewidth]{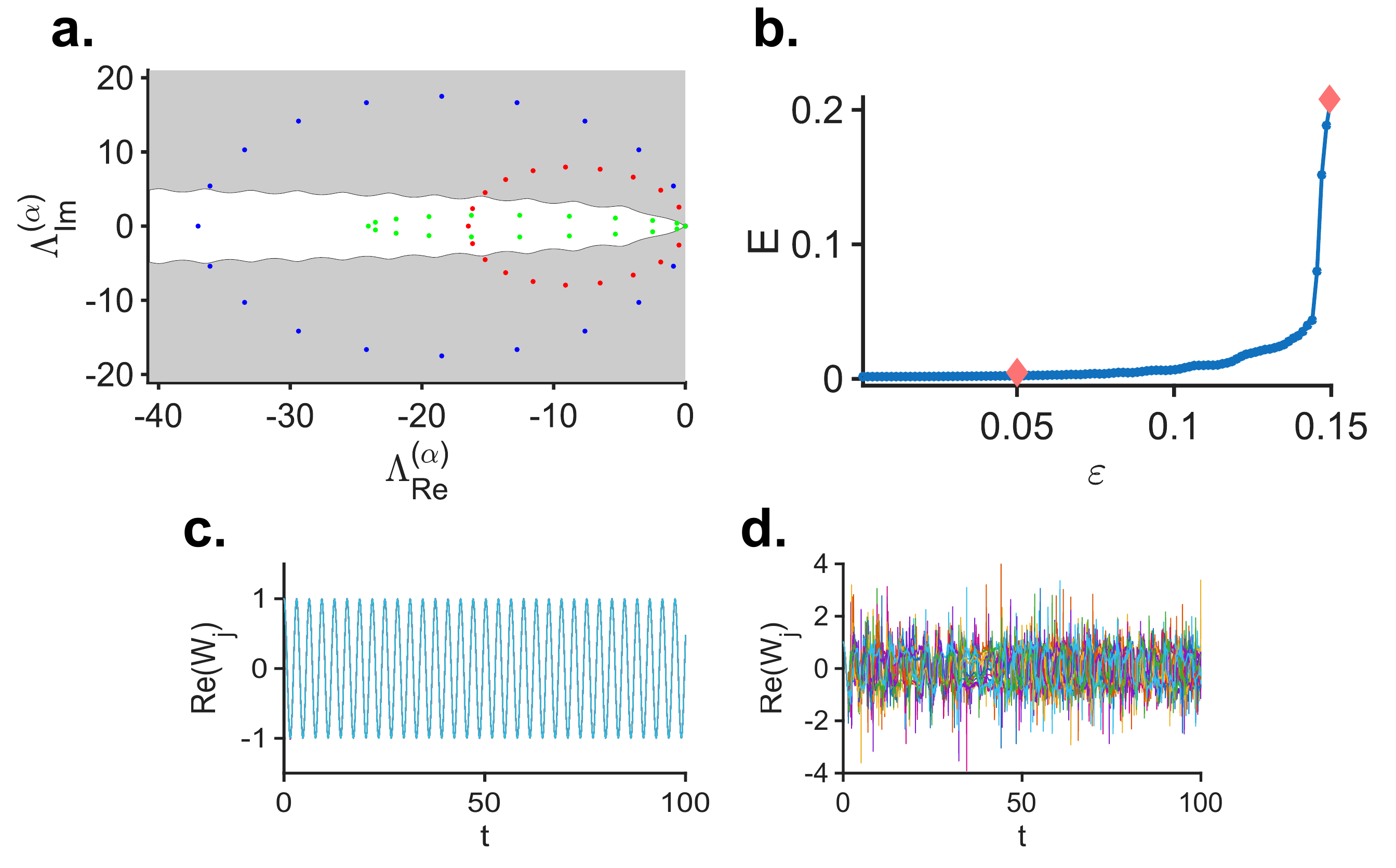}
    \caption{\textbf{Synchronization in time-varying directed SL networks with non-resonant coupling}. The temporal underlying structure is the same used in Fig.~\ref{fig2} as well as the averaged network obtained using $\zeta=0.37$. \textbf{a.} Laplacian spectra of the first configuration (blue), the second configuration (red), and their
    weighted average (green). The grey region denotes $\mathcal{M}(\Lambda)>0$;
    the white region away from the boundary denotes
    $\mathcal{M}(\Lambda)<0$.
    \textbf{c}, Time-averaged synchronization error defined
    in Eq.~\eqref{eq:syncerror}.
    The highlighted points correspond to $\varepsilon=0.05$ and $0.15$.
    \textbf{c} and \textbf{f}, Real parts of the oscillator states over time, respectively  for
    $\varepsilon=0.05$, and for $\varepsilon=0.15$.
    Parameters:
    $\mathbf{h}^{(1)}(w_j)=w_j^5$, $T=1$,
    $\sigma=1$, $\beta=1+2i$, and $\mu=0.2\times(1+1.5i)$.}
    \label{fig:4}
\end{figure}

\begin{figure}[!hb]
    \centering
    \includegraphics[width=1\linewidth]{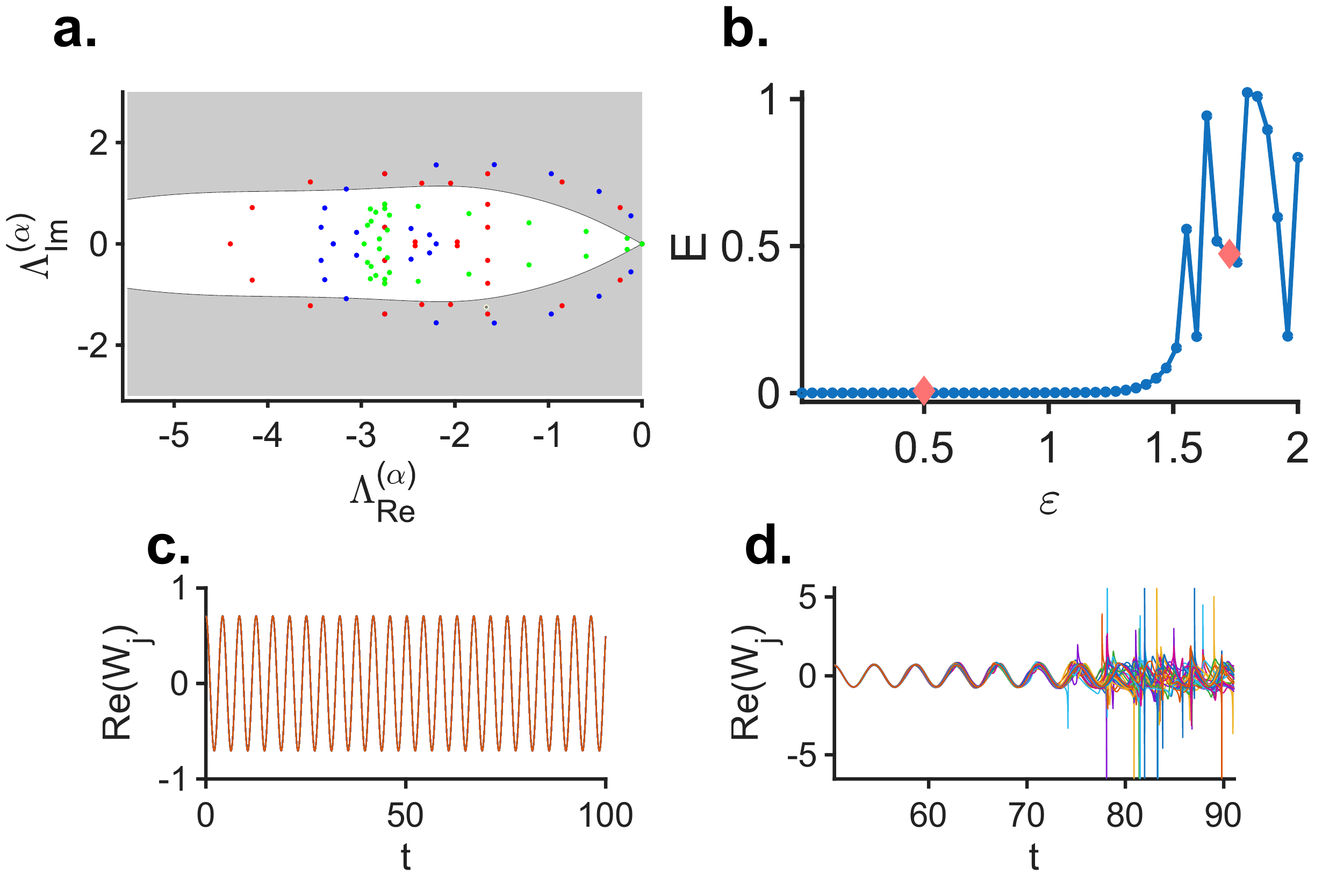}
    \caption{\textbf{Synchronization in time-varying directed SL hypergraphs with non-resonant coupling}. The temporal underlying structure is the same used in Fig.~\ref{fig3} as well as the averaged hypergraph obtained using $\zeta=0.65$. \textbf{a.} generalized Laplacian spectra of the first configuration (blue), the second configuration (red), and their
    weighted average (green). The gray region denotes $\mathcal{M}(\Lambda)>0$;
    the white region away from the boundary denotes
    $\mathcal{M}(\Lambda)<0$.
    \textbf{c}, Time-averaged synchronization error defined
    in Eq.~\eqref{eq:syncerror}.
    The highlighted points correspond to $\varepsilon=0.5$ and $\varepsilon=1.7$.
    \textbf{c} and \textbf{f}, Real parts of the oscillator states over time, respectively  for
    $\varepsilon=0.5$, and for $\varepsilon=1.7$.
    Parameters:
    $\mathbf{h}^{(2)}(w_{j_1},w_{j_2})
=w_{j_1}^{2}\bar{w}_{j_1}w_{j_2}$, $T=2$,
    $\sigma=0.5+2i$, $\beta=1+1i$, and $\mu=0.5\times(1-i)$.}
    \label{fig5}
\end{figure}

\section{Conclusion}
\label{conclusion}
In this work, we investigated global synchronization in time-varying directed higher-order structures, with particular emphasis on the analytical difficulties introduced by the combined effects of directionality, nonlinear coupling, and temporal switching. By considering directed and diffusive higher-order interactions, we have shown that the response of the synchronization manifold to infinitesimal perturbations is governed by a non-autonomous linear system evolving on two distinct time scales, with the fast time dependence carried by the generalized Laplacian matrix. Under the assumption that the underlying directed structures are balanced at all times, the same change of variables can be applied to both the linearized system of the temporal hypergraph and that of its averaged counterpart to isolate perturbations transverse to the synchronization manifold. We then established a sufficient condition for the exponential stability of the synchronization manifold in the temporal hypergraph by showing that, for sufficiently fast switching, the exponential decay of the transverse modes can be inferred from those of the corresponding averaged system. Finally, we illustrated the proposed framework by using a higher-order Stuart--Landau model on directed balanced hyperrings, showing that although directionality may hinder synchronization in the individual static configurations, sufficiently fast switching can restore a stable synchronous state.

Several questions remain open. A natural extension is to investigate temporal higher-order networks with stochastic or irregular switching, as well as regimes in which the switching is slow and the averaged description is no longer sufficient. It would also be valuable to consider continuously evolving structures, where interaction weights are smooth functions of time rather than piecewise constant quantities. Another major challenge is to obtain an analytical estimate of the threshold switching timescale $\varepsilon^*$, which separates the regime in which averaging guarantees synchronization from the regime where the full temporal dynamics must be considered. Finally, the Stuart–Landau hypergraph framework introduced here could be employed to study other collective processes driven by higher-order diffusion, including pattern formation, oscillation death, amplitude chimeras, and the onset of spatiotemporal instabilities.

\section*{Acknowledgment}
W. Segnou acknowledges financial support from the Erasmus+ Programme under Grant Agreement No. 2025-1-BE01-KA171-HED-000340208.

\section*{Competing interests}
The authors declare no competing interests.

\bibliographystyle{unsrt}

\bibliography{Biblio}

\newpage
\onecolumngrid

\appendix

\section{Proof of Lemma~\ref{eq:lemma1}}

Lemma~\ref{eq:lemma1} in the main text can be viewed as a generalization of a similar lemma proposed by Stilwell et al.~\cite{stilwell2006sufficient} and its proof is based on the same main argument. Let us hereby prove Lemma~\ref{eq:lemma1} in full detail.

Since Eq.~\eqref{eq:averagesys} is uniformly exponentially stable, there exist a symmetric matrix ${\bf Q}(t)$ as well as positive scalars $\eta$, $\rho,$ and $\mu$ such that the Lyapunov function 
\begin{equation}
    \label{eq:lyapfunct}
    v({\bm x}(t),t)={\bm x}^\top(t){\bf Q}(t){\bm x}(t)\, ,
\end{equation}
satisfies 
\begin{equation}
\begin{aligned}
\label{eq:cond}
    \eta\|{\bm x}(t)\|^2\;\;&\le  v({\bm x}(t),t)\le \rho\|{\bm x}(t)\|^2,\\
    \frac{d}{dt}\;v({\bm x}(t),t) &\le -\mu \|{\bm x}(t)\|^2\, ,
\end{aligned}
\end{equation}
for all $t$.

Our goal is to prove that, $v({\bm z}(t),t)$ is also a Lyapunov function of Eq.~\eqref{eq:tempsys} for sufficiently small value of $\varepsilon$. This claim is achieved by showing that for sufficiently small value of $\varepsilon$, 
\begin{equation}
    \label{eq:diffenergy}
    \Delta\; v({\bm z}, t+\varepsilon T,t) = v({\bm z}(t+\varepsilon T),t+\varepsilon T)-v({\bm z}(t),t)
\end{equation}
is negative definite for all $t$. From the expression of Lyapunov function defined in Eq.~\eqref{eq:lyapfunct}, we get
\begin{equation}
 \label{eq:diffenergyexpan}
\Delta\; v({\bm z}, t+\varepsilon T,t) = {\bm z}^\top(t+\varepsilon T){\bf Q}(t+\varepsilon T){\bm z}(t+\varepsilon T)-{\bm z}^\top(t){\bf Q}(t){\bm z}(t).
\end{equation}

By denoting ${\bm \Phi}_{{\bf E}}(t,t_0)$ the fundamental matrix corresponding to ${\bf A}(t) +{\bf E}(t/\varepsilon){\bf B}(t)$, i.e., 
\begin{equation}
    \label{eq:soltemporal}
    {\bm z}(t) = {\bm \Phi}_{{\bf E}}(t,t_0){\bm z}_0\, ,
\end{equation}
is a solution of Eq.~\eqref{eq:tempsys}. Then, the variation of the Lyapunov function between $t$ and $t+\varepsilon T$ can be written as
\begin{equation}
\label{eq:diffenergyexpanPhi}
\Delta\; v({\bm z}, t+\varepsilon T,t) = {\bm z}^\top(t)\left[{\bm \Phi}_{{\bf E}}^\top(t+\varepsilon T,t){\bf Q}(t+\varepsilon T){\bm \Phi}_{{\bf E}}(t+\varepsilon T,t)-{\bf Q}(t)\right]{\bm z}(t).    
\end{equation}

Let us measure the error between Eq.~\eqref{eq:averagesys} and Eq.~\eqref{eq:tempsys}. By denoting ${\bm \Phi}_{\langle \bf {E}\rangle}(t,t_0)$ the fundamental matrix of ${\bf A}(t) +\langle {\bf E}\rangle{\bf B}(t)$, we define

\begin{equation}
\label{eq:ErrorEvol}
{\bf H}(t+\varepsilon T, t) = {\bm \Phi}_{{\bf E}}(t+\varepsilon T, t)-{\bm \Phi}_{\langle \bf {E}\rangle}(t+\varepsilon T, t).    
\end{equation}

Denote ${\bf M}_\varepsilon(\sigma_\ell)={\bf A}(\sigma_\ell) +{\bf E}(\sigma_\ell/\varepsilon){\bf B}(\sigma_\ell)$. The Peano-Baker series expansion of ${\bm \Phi}_{{\bf E}}(t+\varepsilon T, t)$ gives us

\begin{equation}
\begin{aligned}
\label{eq:PeanoBakertemp}
{\bm \Phi}_{{\bf E}}(t+\varepsilon T, t) &= {\bf I}+\int_t^{t+\varepsilon T}{\bf M}_\varepsilon(\sigma_1)\; d\sigma_1 + \int_t^{t+\varepsilon T}{\bf M}_\varepsilon(\sigma_1)\; \int_t^{\sigma_1}{\bf M}_\varepsilon(\sigma_2)\;d\sigma_2d\sigma_1 \\ 
&\quad +\int_t^{t+\varepsilon T}{\bf M}_\varepsilon(\sigma_1)\; \int_t^{\sigma_1}{\bf M}_\varepsilon(\sigma_2)\int_t^{\sigma_2}{\bf M}_\varepsilon(\sigma_3)\;d\sigma_3d\sigma_2d\sigma_1+\dots
\end{aligned}
\end{equation}
and similarly for ${\bm \Phi}_{\langle \bf {E}\rangle}(t+\varepsilon T, t)$. By denoting ${\bm \Gamma}$ the difference of the first order terms of the Peano-Baker series of ${\bm \Phi}_{{\bf E}}(t+\varepsilon T, t)$ and ${\bm \Phi}_{\langle \bf {E}\rangle}(t+\varepsilon T, t)$, one obtains
\begin{equation}
\label{eq:difffirstPeanoBaker}
\begin{aligned}
{\bm \Gamma} &= \int_t^{t+\varepsilon T}{\bf M}_\varepsilon(\sigma_1)\; d\sigma_1-\int_t^{t+\varepsilon T}{\bf M}_0(\sigma_1)\; d\sigma_1= \int_t^{t+\varepsilon T}\tilde{\bf {E}}(\sigma_1/\varepsilon){\bf B(\sigma_1)}\; d\sigma_1,
\end{aligned}
\end{equation}
where ${\bf M}_0(\sigma_\ell)={\bf A}(\sigma_\ell) +\langle \bf {E}\rangle{\bf B}(\sigma_\ell)$ and $\tilde{\bf {E}}(t) = {\bf E}(t)-\langle \bf {E}\rangle$. Remark that we have ${\bf B}(\sigma_1)={\bf B}(t)+\left[{\bf B}(\sigma_1)-{\bf B}(t)\right]$. Then,
\begin{equation}
\label{eq:newGamma}
{\bm \Gamma}=\int_t^{t+\varepsilon T}\tilde{\bf {E}}(\sigma_1/\varepsilon)\left[{\bf B(\sigma_1)}-{\bf B}(t)\right]\; d\sigma_1.
\end{equation}
Furthermore, since ${\bf B}(\sigma_1)-{\bf B}(t)=\displaystyle \int_t^{\sigma_1}\frac{d{{\bf B}}}{dr}(r)dr$, we deduce
\begin{equation}
    \|{\bf B}(\sigma_1)-{\bf B}(t)\|\le L_{{\bf B}}(\sigma_1-t).
\end{equation}
Denote $M_{{\bf E}}= \sup\limits_s\|\tilde{\bf {E}}(s)\|$. Then,
\begin{equation}
\label{eq:boundupGamma}
\|{\bm \Gamma}\|\le \frac{1}{2}M_{{\bf E}}L_{{\bf B}}T^2\varepsilon^2=O(\varepsilon^2).
\end{equation}
Define
\begin{equation}
    \alpha
=
\sup_{t}\|\mathbf{A}(t)\|
+
M_{\mathbf B}
\max\left\{
\sup_{t}\|\mathbf{E}(t)\|,
\|\bar{\mathbf E}\|
\right\}.
\end{equation}
We obtain

 \begin{equation}
\begin{aligned}
\|{\bm \Phi}_{{\bf E}}(t+\varepsilon T,t)\|
&\le
1+\alpha\varepsilon T
+
\frac{(\alpha\varepsilon T)^2}{2!}
+
\frac{(\alpha\varepsilon T)^3}{3!}
+\cdots
\\
&\le
e^{\alpha\varepsilon T}.
\end{aligned}
\end{equation}
as well as $
\|{\bm \Phi}_{\langle \bf {E}\rangle}(t+\varepsilon T,t)\|
\le
e^{\alpha\varepsilon T}.
$ Then,
 
\begin{equation}
\begin{aligned}
\|{\bf H}(t+\varepsilon T, t)\|&\le \frac{1}{2}M_{{\bf E}}L_{{\bf B}}T^2\varepsilon^2 + 2(e^{\varepsilon T\alpha}-1-\varepsilon T\alpha),  \\
&\le \left(\frac{1}{2}M_{{\bf E}}L_{{\bf B}}T^2\varepsilon^2+2\alpha^2\varepsilon ^2T^2\right)+O(\varepsilon^3),
\end{aligned}
\end{equation}
so $\|{\bf H}(t+\varepsilon T, t)\|=O(\varepsilon^2)$. Denoting $R(\varepsilon)=M_{{\bf E}}L_{{\bf B}}T^2\varepsilon^2/2 + 2(e^{\varepsilon T\alpha}-1-\varepsilon T\alpha).$

Recall that the following inequalities follow from the uniformly exponential stability of Eq.~\eqref{eq:averagesys}

\begin{equation}
\begin{aligned}
\|{\bf Q}(t)\| &\le \rho, \\
\|\Phi_{\langle \bf {E}\rangle}(t,t_0)\|
&\le
\sqrt{\frac{\rho}{\eta}}\,
e^{-\frac{\mu}{2\rho}(t-t_0)},  \\
 v({\bm x}(t),t)
&\le
e^{-\frac{\mu}{\rho}(t-t_0)}
v({\bm x}(t_0),t_0). 
\end{aligned}
\end{equation}
for $t\ge t_0.$

Noting that ${\bf H}(t+\varepsilon T, t)+{\bm \Phi}_{\langle \bf {E}\rangle}(t+\varepsilon T, t)= {\bm \Phi}_{{\bf E}}(t+\varepsilon T, t)$. We get

\begin{equation}
\begin{aligned}
\Delta v({\bm z},t+\varepsilon T,t)
&=
{\bm z}^{T}(t)
\Big(
\Phi_{\langle \bf {E}\rangle}^{T}(t+\varepsilon T,t)
{\bf Q}(t+\varepsilon T)
\Phi_{\langle \bf {E}\rangle}(t+\varepsilon T,t)
-
{\bf Q}(t)
\Big)
{\bm z}(t)
\\
&\hspace{-1.5cm}+
{\bm z}^{T}(t)
\Big(
\Phi_{\langle \bf {E}\rangle}^{T}(t+\varepsilon T,t)
{\bf Q}(t+\varepsilon T)
{\bf H}(t+\varepsilon T,t)
+
{\bf H}^{T}(t+\varepsilon T,t)
{\bf Q}(t+\varepsilon T)
\Phi_{\langle \bf {E}\rangle}(t+\varepsilon T,t)
\\
&\hspace{-1.5cm}+
{\bf H}^{T}(t+\varepsilon T,t)
{\bf Q}(t+\varepsilon T)
{\bf H}(t+\varepsilon T,t)
\Big)
{\bm z}(t),
\end{aligned}
\end{equation}
and assume ${\bm x}(t)={\bm z}(t)$ at time t. We obtain

\begin{equation}
\begin{aligned}
\eta
\left(
e^{-\mu \varepsilon T/\rho}-1
\right)
\|{\bm z}(t)\|^2
&\ge{\bm z}^{T}(t)
\Big(
\Phi_{\langle \bf {E}\rangle}^{T}(t+\varepsilon T,t)
{\bf Q}(t+\varepsilon T)
\Phi_{\langle \bf {E}\rangle}(t+\varepsilon T,t)
-
{\bf Q}(t)
\Big)
{\bm z}(t)
,\\
 \rho R(\varepsilon)^2\|{\bm z}(t)\|^2&\ge{\bm z}^{T}(t)
\Big(
{\bf H}^{T}(t+\varepsilon T,t)
{\bf Q}(t+\varepsilon T)
{\bf H}(t+\varepsilon T,t)
\Big)
{\bm z}(t), 
\end{aligned}
\end{equation}
as well as

\begin{equation}
    \begin{aligned}
         2\rho
\sqrt{\frac{\rho}{\eta}}
e^{-\mu \varepsilon T/(2\rho)}
R(\varepsilon)\|{\bm z}(t)\|^2&\ge{\bm z}^{T}(t)
\Big(
\Phi_{\langle \bf {E}\rangle}^{T}(t+\varepsilon T,t)
{\bf Q}(t+\varepsilon T)
{\bf H}(t+\varepsilon T,t)
+\\
&+{\bf H}^{T}(t+\varepsilon T,t)
{\bf Q}(t+\varepsilon T)
\Phi_{\langle \bf {E}\rangle}(t+\varepsilon T,t)
\Big)
{\bm z}(t).\\
    \end{aligned}
\end{equation}
We get 
\begin{equation}
    \Delta v({\bm z},t+\varepsilon T,t)\le G(\varepsilon)\|{\bm z}(t)\|^2
\end{equation}
with 
\begin{equation}
G(\varepsilon)=
\eta
\left(
e^{-\mu \varepsilon T/\rho}-1
\right)+
2\rho
\sqrt{\rho/\eta}\;
e^{-\mu \varepsilon T/(2\rho)}
R(\varepsilon)+
\rho R(\varepsilon)^2.
\end{equation}
Since $R(\varepsilon)=O(\varepsilon^{2})$ as $\varepsilon\to 0$ and
$R(0)=0$, it follows that
\begin{equation}
G(0)=0,
\text{ and }
G'(0)
=
-\frac{\mu\eta T}{\rho}
<0.
\end{equation}
Therefore, by continuity of $G$, there exists $\varepsilon^{*}>0$ such
that
$
G(\varepsilon)<0,
\text{ for all }\,\varepsilon\in(0,\varepsilon^{*}).
$
Consequently, for every fixed $\varepsilon\in(0,\varepsilon^{*})$ and
every ${\bm z}\neq 0$,
\begin{equation}
\Delta v({\bm z},t+\varepsilon T,t)<0.
\end{equation}

Fix $\varepsilon\in(0,\varepsilon^{*})$. Since the above estimate is
uniform with respect to the initial time, there exists a constant
$\gamma>0$ such that
\begin{equation}
\begin{aligned}
\Delta v({\bm z},t_{0}+\varepsilon T,t_{0})
&=
v({\bm z}(t_{0}+\varepsilon T),t_{0}+\varepsilon T)
-
v({\bm z}(t_{0}),t_{0})
\\
&\le
-\gamma\|{\bm z}(t_{0})\|^{2},
\end{aligned}
\end{equation}
for all $t_{0}\ge 0$. Using 
$
v({\bm z}(t_{0}),t_{0})
\le
\rho\|{\bm z}(t_{0})\|^{2},
$
we obtain
\begin{equation}
-\gamma\|{\bm z}(t_{0})\|^{2}
\le
-\frac{\gamma}{\rho}
v({\bm z}(t_{0}),t_{0}),
\end{equation}
and hence
\begin{equation}
v({\bm z}(t_{0}+\varepsilon T),t_{0}+\varepsilon T)
\le
\left(
1-\frac{\gamma}{\rho}
\right)
v({\bm z}(t_{0}),t_{0}).
\end{equation}
Defining
$
q
=
1-\gamma/\rho,
\text{ and remark }
0<q<1.
$ Iterating the previous inequality yields
\begin{equation}
v({\bm z}(t_{0}+k\varepsilon T),t_{0}+k\varepsilon T)
\le
q^{k}
v({\bm z}(t_{0}),t_{0}),
\end{equation}
for every integer $k\ge 0$. Using the bounds
$
\eta\|{\bm z}\|^{2}
\le
v({\bm z},t)
\le
\rho\|{\bm z}\|^{2},
$
we further obtain
\begin{equation}
\|{\bm z}(t_{0}+k\varepsilon T)\|
\le
\sqrt{\frac{\rho}{\eta}}\,
q^{k/2}
\|{\bm z}(t_{0})\|.
\end{equation}
Since $0<q<1$, the right-hand side decays exponentially as
$k\to\infty$.

Finally, because the system matrix in Eq.~\eqref{eq:tempsys} is
uniformly bounded, the state transition matrix over each interval of
length at most $\varepsilon T$ is uniformly bounded. Therefore, the
exponential decay obtained at the discrete times
$t_{0}+k\varepsilon T$ extends to all $t\ge t_{0}$. Hence, the trivial
solution of Eq.~\eqref{eq:tempsys} is uniformly exponentially stable.

\section{Linearization around the synchronized Stuart--Landau limit cycle}
\label{App:linearization}

We consider a set of $N$ identical Stuart–Landau oscillators coupled through directed higher-order interactions
 of orders
$d=1,\ldots,D$. For a given order $d$, a directed $d$-hyperedge
consists of one node $\ell$ located at the head of the interaction and
$d$ nodes $j_1,\ldots,j_d$ located in its tail. The latter jointly
influence the head node, while no reciprocal influence is implied by
that hyperedge.

The directed interactions of order $d$ are encoded by the
time-dependent adjacency tensor
$\mathbf{A}^{(d)}(t/\varepsilon)$, whose entries
$\mathbf{A}^{(d)}_{\ell j_1\cdots j_d}(t/\varepsilon)$
measure the strength of the directed interaction from the set of tail
nodes $j_1,\ldots,j_d$ towards the head node $\ell$. We assume that
the tensor is invariant under permutations of its tail indices,
$
\mathbf{A}^{(d)}_{\ell j_1\cdots j_d}
=
\mathbf{A}^{(d)}_{\ell \pi(j_1\cdots j_d)}.
$
Allowing interactions of different orders to coexist, the dynamics of
oscillator $\ell$ is written as
\begin{equation}
\label{eq:general_SL_directed}
\frac{dw_\ell}{dt}
=
\mathbf{f}(w_\ell)
+
\sum_{d=1}^{D}
\mu^{(d)}
\sum_{j_1,\ldots,j_d=1}^{N}
\mathbf{A}^{(d)}_{\ell j_1\cdots j_d}(t/\varepsilon)
\,
\mathbf{g}^{(d)}
\left(
w_\ell,
w_{j_1},\ldots,w_{j_d}
\right),
\end{equation}
where
\begin{equation}
\mathbf{f}(w_\ell)
=
\sigma w_\ell
-
\beta |w_\ell|^2w_\ell
\end{equation}
is the local Stuart--Landau dynamics, with
$\sigma=\sigma_{\Re}+i\sigma_{\Im}$ and
$\beta=\beta_{\Re}+i\beta_{\Im}$, while
$
\mu^{(d)}
=
\mu^{(d)}_{\Re}
+
i\mu^{(d)}_{\Im}
$
denotes the complex coupling strength associated with interactions of
order $d$. Assuming that the coupling functions $\mathbf{g}^{(d)}$ are of the diffusive type, we write the general Stuart--Landau model involving directed
higher-order interactions of orders $d=1,\ldots,D$,

\begin{equation}
\label{eq:general_SL_explicit}
\begin{aligned}
\frac{dw_\ell}{dt}
={}&
\sigma w_\ell
-
\beta |w_\ell|^2w_\ell
+
\sum_{d=1}^{D}
\mu^{(d)}
\sum_{j_1,\ldots,j_d=1}^{N}
\mathbf{A}^{(d)}_{\ell j_1\cdots j_d}(t/\varepsilon)
\times
\left[
\prod_{s=1}^{d}
w_{j_s}^{a_s^{(d)}}
\bar{w}_{j_s}^{\,b_s^{(d)}}
-
w_\ell^{\,\sum_{s=1}^{d}a_s^{(d)}}
\bar{w}_\ell^{\,\sum_{s=1}^{d}b_s^{(d)}}
\right],
\end{aligned}
\end{equation}

where $\ell=1,\ldots,N$ and $a^{(d)}_s$ and $b^{(d)}_s$ are positive integers. For each interaction order $d$, we introduce
the quantities

\begin{equation}
\label{eq:alpha_gamma_d}
\boldsymbol{\alpha}^{(d)}
=
\sum_{s=1}^{d}
\left(
a_s^{(d)}+b_s^{(d)}
\right),
\qquad
\boldsymbol{\gamma}^{(d)}
=
\sum_{s=1}^{d}
\left(
a_s^{(d)}-b_s^{(d)}
\right).
\end{equation}

The quantity $\boldsymbol{\alpha}^{(d)}$ gives the total polynomial degree of the
$d$-th order interaction, whereas $\boldsymbol{\gamma}^{(d)}$ determines its effective
phase rotation.

In the supercritical regime, i.e.,
$
\sigma_{\Re}>0,
\beta_{\Re}>0,
$
an isolated Stuart--Landau oscillator possesses the stable LC
solution
$
w_{\mathrm{LC}}(t)
=
R\,e^{i\omega t},
$ where
$R
=
|w_{\mathrm{LC}}|
=
\sqrt{\sigma_{\Re}/\beta_{\Re}},
\text{ and }
\boldsymbol{\omega}
=
\sigma_{\Im}
-
\beta_{\Im}
\frac{\sigma_{\Re}}{\beta_{\Re}}.
$

We can easily show that this limit cycle is the synchronization manifold of the model Eq.~\eqref{eq:general_SL_explicit}. Indeed, assume
$
w_1(t)=\cdots=w_N(t)=w_{\mathrm{LC}}(t).
$
For every interaction order $d$ and every directed $d$-hyperedge,

\begin{align}
&
\prod_{s=1}^{d}
w_{j_s}^{a_s^{(d)}}
\bar{w}_{j_s}^{\,b_s^{(d)}}
-
w_\ell^{\sum_s a_s^{(d)}}
\bar{w}_\ell^{\,\sum_s b_s^{(d)}}
=
w_{\mathrm{LC}}^{\sum_s a_s^{(d)}}
\bar{w}_{\mathrm{LC}}^{\,\sum_s b_s^{(d)}}
-
w_{\mathrm{LC}}^{\sum_s a_s^{(d)}}
\bar{w}_{\mathrm{LC}}^{\,\sum_s b_s^{(d)}}
=0.
\end{align}

Hence all coupling terms vanish identically on the synchronization
manifold, independently of the tensors
$\mathbf{A}^{(d)}(t/\varepsilon)$, the coupling strengths $\mu^{(d)}$,
and their temporal evolution.

To investigate the linear stability of the synchronous solution, we
introduce small heterogeneous perturbations of its amplitude and phase,
$
w_\ell(t)
=
w_{\mathrm{LC}}(t)
\left(
1+\rho_\ell(t)
\right)
e^{i\theta_\ell(t)},
$
with
$
|\rho_\ell|\ll1,
|\theta_\ell|\ll1.
$. Linearizing Eq.~\eqref{eq:general_SL_explicit} in the vicinity of the limit cycle yields the following system

\begin{equation}
\label{eq:rho_theta_before_coupling}
\left\{
\begin{aligned}
\dot{\rho}_\ell
&=
-2\sigma_{\Re}\rho_\ell
+
R^{-1}
\Re
\left[
e^{-i\omega t}\mathscr{L}_\ell
\right],
\\[1mm]
\dot{\theta}_\ell
&=
-2\beta_{\Im}
\frac{\sigma_{\Re}}{\beta_{\Re}}
\rho_\ell
+
R^{-1}
\Im
\left[
e^{-i\omega t}\mathscr{L}_\ell
\right],
\end{aligned}
\right.
\end{equation}

where $\mathscr{L}_\ell$ denotes the linear contribution of the coupling term. Let's determine its expression.
one has, 
\begin{equation}
\begin{aligned}
w_{j_s}^{a_s^{(d)}}
\bar{w}_{j_s}^{\,b_s^{(d)}}
={}
R^{(a^{(d)}_s+b^{(d)}_s)}
e^{i(a^{(d)}_s-b^{(d)}_s)\omega t}
\times
\left[
1
+
(a^{(d)}_s+b^{(d)}_s)\rho_{j_s}
+
i(a^{(d)}_s-b^{(d)}_s)\theta_{j_s}
\right]
+
\mathcal{O}(2).
\end{aligned}
\end{equation}

Multiplying the $d$ contributions and neglecting all products of
perturbations gives

\begin{equation}
\label{eq:linear_hd}
\begin{aligned}
\prod_{s=1}^{d}
w_{j_s}^{a_s^{(d)}}
\bar{w}_{j_s}^{\,b_s^{(d)}}
={}&
R^{\boldsymbol{\alpha}^{(d)}}
e^{i\boldsymbol{\gamma}^{(d)}\omega t}
\times
\left[
1
+
\sum_{s=1}^{d}
(a^{(d)}_s+b^{(d)}_s)\rho_{j_s}
+
i
\sum_{s=1}^{d}
(a^{(d)}_s-b^{(d)}_s)\theta_{j_s}
\right]
+
\mathcal{O}(2).
\end{aligned}
\end{equation}

When all arguments are evaluated at the head node $\ell$, the same
calculation yields

\begin{equation}
\label{eq:linear_head_hd}
\begin{aligned}
w_\ell^{\sum_s a_s^{(d)}}
\bar{w}_\ell^{\,\sum_s b_s^{(d)}}
={}&
R^{\boldsymbol{\alpha}^{(d)}}
e^{i\boldsymbol{\gamma}^{(d)}\omega t}
\times
\left[
1+\boldsymbol{\alpha}^{(d)}\rho_\ell
+i\boldsymbol{\gamma}^{(d)}\theta_\ell
\right]
+
\mathcal{O}(2).
\end{aligned}
\end{equation}

Hence, the linearized diffusive contribution associated with a single
directed $d$-hyperedge is

\begin{equation}
\label{eq:linear_diff_d}
\begin{aligned}
&\prod_{s=1}^{d}
w_{j_s}^{a_s^{(d)}}
\bar{w}_{j_s}^{\,b_s^{(d)}}
-
w_\ell^{\sum_s a_s^{(d)}}
\bar{w}_\ell^{\,\sum_s b_s^{(d)}}
=
R^{\boldsymbol{\alpha}^{(d)}}
e^{i\boldsymbol{\gamma}^{(d)}\omega t}
\Bigg[
\sum_{s=1}^{d}
(a^{(d)}_s+b^{(d)}_s)
\left(
\rho_{j_s}-\rho_\ell
\right)
+i
\sum_{s=1}^{d}
(a^{(d)}_s-b^{(d)}_s)
\left(
\theta_{j_s}-\theta_\ell
\right)
\Bigg]
+
\mathcal{O}(2).
\end{aligned}
\end{equation}

The permutation symmetry of the tail indices of
$\mathbf{A}^{(d)}$ now plays a crucial role. Define the
contraction of the adjacency tensor by

\begin{equation}
\label{eq:tensor_contraction}
\mathcal{A}_{\ell j}^{(d)}
\left(t/\varepsilon\right)
=
\sum_{j_2,\ldots,j_d=1}^{N}
\mathbf{A}^{(d)}_{\ell j j_2\cdots j_d}
\left(t/\varepsilon\right).
\end{equation}
Because the adjacency tensor is invariant under any permutation of its
$d$ tail indices, the contraction is independent of the position
occupied by the selected node $j$. Hence, for any vector
$\mathbf{x}=(x_1,\ldots,x_N)^\top$, one has

\begin{equation}
\label{eq:tensor_reduction_identity}
\begin{aligned}
&
\sum_{j_1,\ldots,j_d}
\mathbf{A}^{(d)}_{\ell j_1\cdots j_d}
\sum_{s=1}^{d}
c_s x_{j_s}
=
\left(
\sum_{s=1}^{d}c_s
\right)
\sum_{j=1}^{N}
\mathcal{A}_{\ell j}^{(d)}x_j.
\end{aligned}
\end{equation}

Applying Eq.~\eqref{eq:tensor_reduction_identity} to the amplitude and
phase perturbations gives

\begin{align}
&
\sum_{j_1,\ldots,j_d}
\mathbf{A}^{(d)}_{\ell j_1\cdots j_d}
\sum_{s=1}^{d}
(a^{(d)}_s+b^{(d)}_s)\rho_{j_s}
=
\boldsymbol{\alpha}^{(d)}
\sum_j
\mathcal{A}_{\ell j}^{(d)}
\rho_j,
\\
&
\sum_{j_1,\ldots,j_d}
\mathbf{A}^{(d)}_{\ell j_1\cdots j_d}
\sum_{s=1}^{d}
(a^{(d)}_s-b^{(d)}_s)\theta_{j_s}
=
\boldsymbol{\gamma}^{(d)}
\sum_j
\mathcal{A}_{\ell j}^{(d)}
\theta_j.
\end{align}

We consider the generalized in-degree associated with interactions of
order $d$ as

\begin{equation}
\label{eq:kin_d}
k_{\ell}^{\mathrm{in},(d)}
=
\frac{1}{d!}
\sum_{j_1,\ldots,j_d=1}^{N}
\mathbf{A}^{(d)}_{\ell j_1\cdots j_d}.
\end{equation}

The factor $d!$ accounts for the $d!$ equivalent permutations of the
tail nodes contained in the symmetric adjacency tensor. Consequently,

\begin{equation}
\sum_{j_1,\ldots,j_d}
\mathbf{A}^{(d)}_{\ell j_1\cdots j_d}
=
d!\,
k_{\ell}^{\mathrm{in},(d)}.
\end{equation}

Similarly, we consider

\begin{equation}
\label{eq:kij_d}
k_{\ell j}^{(d)}
=
\frac{1}{(d-1)!}
\sum_{j_2,\ldots,j_d=1}^{N}
\mathbf{A}^{(d)}_{\ell j j_2\cdots j_d},
\qquad
\ell\neq j,
\end{equation}

denoting the number of directed hyperedges containing node $\ell$ in the head and $j$ in the tail. It follows from Eq.~\eqref{eq:tensor_contraction} that

\begin{equation}
\mathcal{A}_{\ell j}^{(d)}
=
(d-1)!
k_{\ell j}^{(d)},
\qquad
\ell\neq j.
\end{equation}

Assuming that the head node does not simultaneously belong to the tail
of the same hyperedge, the generalized Laplacian matrix associated with
interactions of order $d$ is therefore naturally defined by

\begin{equation}
\label{eq:Ld_tensor}
\mathbf{L}^{(d)}_{\ell j}
=
\begin{cases}
\displaystyle
\sum_{j_2,\ldots,j_d}
\mathbf{A}^{(d)}_{\ell j j_2\cdots j_d},
& \ell\neq j,
\\[4mm]
\displaystyle
-\sum_{j_1,\ldots,j_d}
\mathbf{A}^{(d)}_{\ell j_1\cdots j_d},
& \ell=j.
\end{cases}
\end{equation}

Equivalently, in terms of the generalized degrees,

\begin{equation}
\label{eq:Ld_degree}
\mathbf{L}^{(d)}_{\ell j}
=
\begin{cases}
(d-1)!\,k_{\ell j}^{(d)},
& \ell\neq j,
\\[2mm]
-d!\,k_{\ell}^{\mathrm{in},(d)},
& \ell=j.
\end{cases}
\end{equation}

By construction,
$
\sum_{j=1}^{N}
\mathbf{L}^{(d)}_{\ell j}
=0,
$ so that the uniform vector belongs to the right kernel of every
$\mathbf{L}^{(d)}$.

Using Eq.~\eqref{eq:Ld_tensor}, the contribution of interactions of
order $d$ to the linearized coupling can be rewritten as

\begin{equation}
\label{eq:Cd_linear_L}
\begin{aligned}
\mathscr{L}_\ell^{(d)}
={}&
\mu^{(d)}
R^{\boldsymbol{\alpha}^{(d)}}
e^{i\boldsymbol{\gamma}^{(d)}\omega t}
\Bigg[
\boldsymbol{\alpha}^{(d)}
\sum_{j=1}^{N}
\mathbf{L}^{(d)}_{\ell j}(t/\varepsilon)
\rho_j
+
i\boldsymbol{\gamma}^{(d)}
\sum_{j=1}^{N}
\mathbf{L}^{(d)}_{\ell j}(t/\varepsilon)
\theta_j
\Bigg].
\end{aligned}
\end{equation}

Summing over all interaction orders yields

\begin{equation}
\label{eq:Ctotal_linear}
\begin{aligned}
R^{-1}e^{-i\omega t}\mathscr{L}_\ell
=
\sum_{d=1}^{D}
&R^{\boldsymbol{\alpha}^{(d)}-1}
\mu^{(d)}
e^{i(\boldsymbol{\gamma}^{(d)}-1)\omega t}
\times
\left[
\boldsymbol{\alpha}^{(d)}
\sum_j
\mathbf{L}^{(d)}_{\ell j}
\rho_j
+
i\boldsymbol{\gamma}^{(d)}
\sum_j
\mathbf{L}^{(d)}_{\ell j}
\theta_j
\right].
\end{aligned}
\end{equation}

Let

\begin{align}
T_{1}^{(d)}(t)
&=
\mu^{(d)}_{\Re}
\cos\!\left[
(\boldsymbol{\gamma}^{(d)}-1)\omega t
\right]
-
\mu^{(d)}_{\Im}
\sin\!\left[
(\boldsymbol{\gamma}^{(d)}-1)\omega t
\right],
\\
T_{2}^{(d)}(t)
&=
\mu^{(d)}_{\Im}
\cos\!\left[
(\boldsymbol{\gamma}^{(d)}-1)\omega t
\right]
+
\mu^{(d)}_{\Re}
\sin\!\left[
(\boldsymbol{\gamma}^{(d)}-1)\omega t
\right].
\end{align}

The complete linearized amplitude--phase dynamics is consequently

\begin{equation}
\label{eq:full_linear_D}
\begin{aligned}
\begin{cases} \displaystyle\frac{d}{dt}\rho_\ell
=
-2\sigma_{\Re}\rho_\ell
+
\sum_{d=1}^{D}
R^{\boldsymbol{\alpha}^{(d)}-1}
\Bigg[
\boldsymbol{\alpha}^{(d)} T_{1}^{(d)}(t)
\sum_j
\mathbf{L}^{(d)}_{\ell j}(t/\varepsilon)\rho_j
-
\boldsymbol{\gamma}^{(d)} T_{2}^{(d)}(t)
\sum_j
\mathbf{L}^{(d)}_{\ell j}(t/\varepsilon)\theta_j
\Bigg],
\\[2mm]
\displaystyle\frac{d}{dt}\theta_\ell
=
-2\beta_{\Im}
\displaystyle\frac{\sigma_{\Re}}{\beta_{\Re}}
\rho_\ell
+
\sum_{d=1}^{D}
R^{\boldsymbol{\alpha}^{(d)}-1}
\Bigg[
\boldsymbol{\alpha}^{(d)} T_{2}^{(d)}(t)
\sum_j
\mathbf{L}^{(d)}_{\ell j}(t/\varepsilon)\rho_j
+
\boldsymbol{\gamma}^{(d)} T_{1}^{(d)}(t)
\sum_j
\mathbf{L}^{(d)}_{\ell j}(t/\varepsilon)\theta_j
\Bigg].
\end{cases}
\end{aligned}
\end{equation}






Introducing

\begin{equation}
\begin{aligned}
\mathbf{J}_0
=
\begin{pmatrix}
-2\sigma_{\Re} & 0\\[1mm]
-2\beta_{\Im}\dfrac{\sigma_{\Re}}{\beta_{\Re}} & 0
\end{pmatrix}, 
\text{ and }\mathbf{J}_{h}^{(d)}(t)
=
R^{\boldsymbol{\alpha}^{(d)}-1}
\begin{pmatrix}
\boldsymbol{\alpha}^{(d)}T^{(d)}_{1}(t) & -\boldsymbol{\gamma}^{(d)}T^{(d)}_{2}(t)\\
\boldsymbol{\alpha}^{(d)}T^{(d)}_{2}(t) & \boldsymbol{\gamma}^{(d)}T^{(d)}_{1}
\end{pmatrix}.
\end{aligned}
\end{equation}

The complete perturbation dynamics can finally be written as

\begin{equation}
\label{eq:general_compact_linear}
\frac{d\mathcal{E}}{dt}
=
\left[
\mathbf{I}_N\otimes\mathbf{J}_0
+
\sum_{d=1}^{D}
\mathbf{L}^{(d)}(t/\varepsilon)
\otimes
\mathbf{J}_{h}^{(d)}(t)
\right]
\mathcal{E}.
\end{equation}

For a fixed $d$, we thus obtain Eq.~\eqref{eq:SLLinearmodelcompactformhypo}.

\section{Derivation of the instability domain in the complex plane}
\label{App}

We provide here the analytical steps leading to the instability domain used for resonant case. Consider a static directed $d$-uniform hypergraph, so that the generalized Laplacian entering Eq.~\eqref{eq:SLLinearmodelcompactformhypo} is time independent,
\[
\mathbf{L}^{(d)}(t/\varepsilon)=\mathbf{L}^{(d)}.
\]
We assume that $\mathbf{L}^{(d)}$ is diagonalizable and denote its eigenpairs by
$
\mathbf{L}^{(d)}\Psi^{(k)}
=
\Lambda_k\Psi^{(k)},
\text{ with }
k=1,\ldots,N.
$
The amplitude and phase perturbations can then be decomposed along the corresponding right eigenvectors as

\begin{equation}
\label{eq:projeigenb}
\rho_j
=
\sum_{k=1}^{N}
\widehat{\rho}_{k}\Psi_j^{(k)},
\qquad
\theta_j
=
\sum_{k=1}^{N}
\widehat{\theta}_{k}\Psi_j^{(k)}.
\end{equation}

Inserting Eq.~\eqref{eq:projeigenb} into
Eq.~\eqref{eq:SLLinearmodelcompactformhypo} separates the perturbation dynamics into $N$ independent modal equations. For each Laplacian mode $k$, one obtains

\begin{equation}
\label{eq:SLLinearmodelGraphTimDep}
\frac{d}{dt}
\begin{pmatrix}
\widehat{\rho}_{k}\\
\widehat{\theta}_{k}
\end{pmatrix}
=
\mathbf{J}_{k}
\begin{pmatrix}
\widehat{\rho}_{k}\\
\widehat{\theta}_{k}
\end{pmatrix},
\qquad
\mathbf{J}_{k}
=
\mathbf{J}_0
+
\Lambda_k\mathbf{J}_{\mathbf h},
\end{equation}
where the local and coupling contributions are respectively given by
\begin{equation}
\label{eq:J0hautb}
\mathbf{J}_0
=
\begin{pmatrix}
-2\sigma_{\Re} & 0\\[1mm]
-2\beta_{\Im}\dfrac{\sigma_{\Re}}{\beta_{\Re}} & 0
\end{pmatrix},
\qquad
\mathbf{J}_{\mathbf h}
=
|w_{\text{LC}}|^{\boldsymbol{\alpha}-1}
\begin{pmatrix}
\boldsymbol{\alpha}\mu_{\Re} & -\mu_{\Im}\\
\boldsymbol{\alpha}\mu_{\Im} & \mu_{\Re}
\end{pmatrix}.
\end{equation}

Accordingly, the stability of the synchronized state can be determined by studying the spectrum of the $2\times2$ matrix $\mathbf{J}_{k}$ for every transverse Laplacian mode. Because the hypergraph is directed, the eigenvalues of $\mathbf{L}^{(d)}$ are not necessarily real. We therefore write
\[
\Lambda_k
=
\Lambda_{\Re}^{(k)}
+
i\Lambda_{\Im}^{(k)}.
\]

To simplify the expressions that follow, we define

\begin{equation}
\label{eq:newvar}
\begin{aligned}
\varphi
&=
(\boldsymbol{\alpha}+1)
|w_{\text{LC}}|^{\boldsymbol{\alpha}-1}
\mu_{\Re},
\\[1mm]
\gamma
&=
-2\sigma_{\Re}
|w_{\text{LC}}|^{\boldsymbol{\alpha}-1}
\left(
\mu_{\Re}
+
\frac{\beta_{\Im}}{\beta_{\Re}}\mu_{\Im}
\right),
\\[1mm]
\eta
&=
\boldsymbol{\alpha}
\left(
\mu_{\Re}^{2}+\mu_{\Im}^{2}
\right)
|w_{\text{LC}}|^{2(\boldsymbol{\alpha}-1)}.
\end{aligned}
\end{equation}

With these quantities, the real and imaginary parts of the trace of $\mathbf{J}_{k}$ become

\begin{align}
\left(
\operatorname{tr}\mathbf{J}_{k}
\right)_{\Re}
&=
-2\sigma_{\Re}
+
\varphi\Lambda_{\Re}^{(k)},
\label{eq:trace_real}
\\
\left(
\operatorname{tr}\mathbf{J}_{k}
\right)_{\Im}
&=
\varphi\Lambda_{\Im}^{(k)}.
\label{eq:trace_imag}
\end{align}

Similarly, the determinant satisfies

\begin{align}
\left(
\det\mathbf{J}_{k}
\right)_{\Re}
&=
\gamma\Lambda_{\Re}^{(k)}
+
\eta
\left[
\left(\Lambda_{\Re}^{(k)}\right)^2
-
\left(\Lambda_{\Im}^{(k)}\right)^2
\right],
\label{eq:det_real}
\\
\left(
\det\mathbf{J}_{k}
\right)_{\Im}
&=
\gamma\Lambda_{\Im}^{(k)}
+
2\eta
\Lambda_{\Re}^{(k)}
\Lambda_{\Im}^{(k)}.
\label{eq:det_imag}
\end{align}

The two growth rates associated with mode $k$ are the roots of
\[
\lambda^2
-
\left(
\operatorname{tr}\mathbf{J}_{k}
\right)\lambda
+
\det\mathbf{J}_{k}
=
0.
\]
Following the complex-eigenvalue treatment introduced in
Refs.~~\cite{asllani2014theory,di2017benjamin}, let

\begin{align}
\mathcal{A}_{k}
&=
\left[
\left(
\operatorname{tr}\mathbf{J}_{k}
\right)_{\Re}
\right]^2
-
\left[
\left(
\operatorname{tr}\mathbf{J}_{k}
\right)_{\Im}
\right]^2
-
4
\left(
\det\mathbf{J}_{k}
\right)_{\Re},
\\
\mathcal{B}_{k}
&=
2
\left(
\operatorname{tr}\mathbf{J}_{k}
\right)_{\Re}
\left(
\operatorname{tr}\mathbf{J}_{k}
\right)_{\Im}
-
4
\left(
\det\mathbf{J}_{k}
\right)_{\Im}.
\end{align}

The real and imaginary parts of the square root of
$\mathcal{A}_{k}+i\mathcal{B}_{k}$ can be written as

\begin{equation}
\varpi_{k}
=
\sqrt{
\frac{
\mathcal{A}_{k}
+
\sqrt{
\mathcal{A}_{k}^{2}
+
\mathcal{B}_{k}^{2}
}
}{2}
},
\end{equation}

and

\begin{equation}
\delta_{k}
=
\operatorname{sign}
\left(
\mathcal{B}_{k}
\right)
\sqrt{
\frac{
-\mathcal{A}_{k}
+
\sqrt{
\mathcal{A}_{k}^{2}
+
\mathcal{B}_{k}^{2}
}
}{2}
}.
\end{equation}

when $\mathcal{B}_k\neq0$.
For $\mathcal{B}_k=0$, we set $\delta_k=\sqrt{-\mathcal{A}_k}$ if $\mathcal{A}_k<0$,
and $\delta_k=0$ otherwise. Consequently, the eigenvalue of the modal Jacobian with maximal real part reads

\begin{equation}
\label{eq:eigenvalues}
\lambda_{k}
=
\frac{1}{2}
\left[
\left(
\operatorname{tr}\mathbf{J}_{k}
\right)_{\Re}
+
\varpi_{k}
\right]
+
\frac{i}{2}
\left[
\left(
\operatorname{tr}\mathbf{J}_{k}
\right)_{\Im}
+
\delta_{k}
\right].
\end{equation}

The Master Stability Function is defined as 
\begin{equation}
\label{eq:explicit_msf}
\mathcal{M}
\left(
\Lambda_k
\right)
=
\frac{1}{2}
\left[
\left(
\operatorname{tr}\mathbf{J}_{k}
\right)_{\Re}
+
\varpi_{k}
\right].
\end{equation}

The synchronized limit cycle is therefore linearly asymptotically stable whenever

\begin{equation}
\label{eq:condstab}
\mathcal{M}
\left(
\Lambda_k
\right)
<0,
\qquad
k=2,\ldots,N.
\end{equation}

The mode $k=1$, associated with
$\Lambda_1=0$, is not included in this condition because it is tangent to the synchronization manifold and contains the neutral collective phase direction. We note that the condition in Eq.~\eqref{eq:condstab} is not guaranteed when
$
\left(\operatorname{tr}\mathbf{J}_{k}\right)_{\Re}>0.
$
Therefore, a necessary condition is to impose
\begin{equation}
\mu_{\Re}>
-\frac{2\sigma_{\Re}}
{(\boldsymbol{\alpha}+1)|w_{\text{LC}}|^{\boldsymbol{\alpha}-1}\Lambda_*},
\qquad
\Lambda_*=\max_{k=2,\ldots,N}|\Re(\Lambda_k)|.
\end{equation}

An explicit description of the instability domain can be obtained by eliminating the square roots appearing in
Eq.~\eqref{eq:explicit_msf}. Let
$
x=\Lambda_{\Re},\;
y=\Lambda_{\Im},
$
and introduce the two real polynomials

\begin{equation}
\label{eq:S1S2}
S_1(x)
=
C_{12}x^2
+
C_{11}x
+
C_{10},
\qquad
S_2(x)
=
C_{24}x^4
+
C_{23}x^3
+
C_{22}x^2
+
C_{21}x.
\end{equation}

Their coefficients depend only of the Stuart--Landau and coupling parameters and are given by

\begin{align}
C_{24}
&=
\varphi^2\eta,
&
C_{12}
&=
4\eta^2-\varphi^2\eta,
\\
C_{23}
&=
\varphi^2\gamma
-
4\varphi\sigma_{\Re}\eta,
&
C_{11}
&=
4\gamma\eta
-
\varphi^2\gamma,
\\
C_{22}
&=
4\sigma_{\Re}^2\eta
-
4\varphi\sigma_{\Re}\gamma,
&
C_{10}
&=
\gamma^2
+
2\varphi\sigma_{\Re}\gamma
+
4\sigma_{\Re}^2\eta,
\\
C_{21}
&=
4\sigma_{\Re}^2\gamma.
\end{align}

After straightforward algebra, the condition
$\mathcal{M}(\Lambda)>0$ is equivalent to

\begin{equation}
\label{eq:condS2S1}
S_2(x)
<
y^2S_1(x).
\end{equation}

Equation~\eqref{eq:condS2S1} therefore identifies the instability region in the complex plane. The synchronized solution becomes unstable whenever at least one nonzero generalized Laplacian eigenvalue $\Lambda_k
=
\Lambda_{\Re}^{(k)}
+
i\Lambda_{\Im}^{(k)}$ belongs to this domain. Conversely, global synchronization is linearly asymptotically stable when the entire transverse Laplacian spectrum lies outside it.

\end{document}